\documentclass{article}
\usepackage{graphicx}
\usepackage{subfig}
\usepackage{amsfonts}
\usepackage{amsmath}
\usepackage{amssymb}
\usepackage{url}
\usepackage{fancyhdr}
\usepackage{indentfirst}
\usepackage{enumerate}
\usepackage{longtable}
\usepackage{float} 

\usepackage{dsfont}
\usepackage[colorlinks=true,citecolor=blue]{hyperref}
\usepackage{amsthm}
\usepackage{color}
\usepackage{natbib}
\usepackage{comment}
\usepackage{capt-of}
\usepackage{multirow}

\def\d{\mathrm{d}}

\newcommand{\LT}{L_T^g}
\newcommand{\D}{\mathcal {D}}

\newcommand{\cov}{\mathrm{Cov}}

\newcommand{\E}{\mathbb{E}}

\newcommand{\F}{\mathcal{F}}

\newcommand{\R}{\mathbb{R}}

\newcommand{\p}{\mathbb{P}}

\newcommand{\id}{\mathds{1}}

\renewcommand{\(}{\left(}
\renewcommand{\)}{\right)}
\renewcommand{\[}{\left[}
\renewcommand{\]}{\right]}

\renewcommand{\geq}{\geqslant}
\renewcommand{\leq}{\leqslant}
\renewcommand{\epsilon}{\varepsilon}

\renewcommand{\cdots}{\dots}

\theoremstyle{plain}
\newtheorem{theorem}{Theorem}

\newtheorem{lemma}{Lemma}
\newtheorem{proposition}{Proposition}
\theoremstyle{definition}
\newtheorem{definition}{Definition}
\newtheorem{example}{Example}

\theoremstyle{remark}
\newtheorem{remark}{Remark}

\newcommand{\cet}{\begin{center}}
	\newcommand{\ecet}{\end{center}}

\usepackage{refcount}
\setrefcountdefault{1}

\newcounter{saveexample}

\newenvironment{examplec}[1][]{%
	\setcounter{saveexample}{\value{example}}%
	\setcounterref{example}{#1}%
	\addtocounter{example}{-1}%
	\begin{example}[cont.]}{\end{example}%
	\setcounter{example}{\thesaveexample}%
}

\allowdisplaybreaks[4]
\usepackage{setspace}

\begin{document}
\sloppy 

\title{A Unified Formula of the Optimal Portfolio for Piecewise Hyperbolic Absolute Risk Aversion Utilities}

\author{Zongxia Liang\thanks{\scriptsize Department of Mathematical Sciences, Tsinghua University, Beijing 100084, China. Email: \texttt{liangzongxia@mail.tsinghua.edu.cn}}
	\and Yang Liu\thanks{\scriptsize Corresponding author. 
		Division of Mathematics, School of Science and Engineering, The Chinese University of Hong Kong, Shenzhen, Guangdong 518172, China. Email: \texttt{yangliu16@cuhk.edu.cn}}		
	\and Ming Ma\thanks{\scriptsize Hangzhou Higgs Asset Management Co., Ltd., Zhejiang 310000, China. Email: \texttt{maming@higgsasset.com} }
	\and Rahul Pothi Vinoth\thanks{\scriptsize Department of Economics, UC Berkeley, CA 94720, US. Email: \texttt{rvinoth@berkeley.edu} }
}

\date{}



\maketitle

\begin{abstract}	
	
	
	We propose a general family of piecewise hyperbolic absolute risk aversion (PHARA) utilities, including many classic and non-standard utilities as examples. A typical application is the composition of a HARA preference and a piecewise linear payoff in asset allocation. We derive a unified closed-form formula of the optimal portfolio, which is a four-term division. The formula has clear economic meanings, reflecting the behavior of risk aversion, risk seeking, loss aversion and first-order risk aversion. We conduct a general asymptotic analysis to the optimal portfolio, which directly serves as an analytical tool for financial analysis. We compare this PHARA portfolio with those of other utility families both analytically and numerically. One main finding is that risk-taking behaviors are greatly increased by non-concavity and reduced by non-differentiability of the PHARA utility. Finally, we use financial data to test the performance of the PHARA portfolio in the market.

	
	\vskip 10pt  
	
	\noindent
	{\bf MSC(2010)}: 91B16, 91G10.
	
	\noindent
	{\bf JEL Classifications}: G11, C61.
	\vskip 10pt

	\noindent
	{\bf Keywords:}  Utility theory; 
	Portfolio selection; 
	Non-differentiability; Asymptotic analysis; Empirical study.
\end{abstract}

\section{Introduction}\label{sec:intro}

In the field of continuous-time portfolio selection, after the seminal work of \cite{M1969} investigates hyperbolic absolute risk aversion (HARA) utilities, more and more non-standard (e.g., non-concave, non-differentiable) utility functions appear. 
These utility functions come from the following optimization problems (where $U$ is the utility function, $\boldsymbol{\pi} \in \mathcal{V}$ is the portfolio process to be selected and $X_T$ is the corresponding terminal wealth): 
\begin{enumerate}[(1)]
	\item The decision maker solves a utility maximization problem with the basic HARA utility: 
	\begin{equation}\label{prob:merton_basic}
		\max_{\boldsymbol{\pi} \in \mathcal{V}} \E[U(X_T)], 
	\end{equation}
	where $U$ is a classic HARA utility function (\cite{M1969}). This is revisited in Example \ref{ex:M}.
	
	\item In the field of asset allocation (e.g., the hedge fund),  $\widehat{U}$ is a HARA preference function and $\Theta$ is usually a (non-concave) piecewise linear payoff function. The decision maker solves a utility maximization problem with the objective being a composition function $U := \widehat{U} \circ \Theta$: 
	\begin{equation}
		\max_{\boldsymbol{\pi}  \in \mathcal{V}} \E[U(X_T)] = \max_{\boldsymbol{\pi}  \in \mathcal{V}} \E[\widehat{U}(\Theta(X_T))]. 
	\end{equation}
	We take \cite{C2000} as an example, which is revisited in Example \ref{ex:C}. See also \cite{BKP2004,L2005,HJ2007,BS2014,HLLM2019}. 
	
	\item In the setting of behavioral finance, the decision maker solves a utility maximization problem with a non-concave utility (e.g., S-shaped):
	\begin{equation}
		\max_{\boldsymbol{\pi} \in \mathcal{V}} \E[U(X_T)], 
	\end{equation}
	where $U$ is a S-shaped function or a composition function related to an S-shaped function. In the above category (2), some papers also use a S-shaped utility. These papers may also be included in this category. We take \cite{LSW2017} as an example and revisit it in Example \ref{ex:LSW}. See also \cite{KZ2007,DZ2020,LL2020}.
	
	\item In the setting of risk management, the decision maker solves the utility maximization problem with other objectives and constraints (e.g., the constraints of liquidation, the performance ratio and the Value-at-Risk): 
	\begin{equation}\label{prob:constraint}
		\begin{aligned}
			& \max_{\boldsymbol{\pi} \in \mathcal{V}} \E[U(X_T^{\boldsymbol{\pi}} )],\\
			& \text{s.t. }  \mathcal{R}(X_T) \geq 0,
		\end{aligned}
	\end{equation}
	where $\mathcal{R}$ is a risk constraint. We take \cite{HK2018} as an example and revisit it in Example \ref{ex:HK}. See also \cite{BVY2019,CHN2019,BBK2021,GLX2023}. 
\end{enumerate} 
A specific utility model may be classified in multiple categories above. Although fortunately the problems of these utilities admit closed-form solutions in many cases, the optimal portfolios are expressed in different forms in the literature, making the analytical analysis difficult to conduct consistently. In this paper, we propose a large family of utility, ``piecewise HARA (PHARA) utility", to unify and contain these classic and non-standard utility functions. Intuitively, the PHARA utility is a direct generalization of the HARA utility: it reduces to a HARA utility on each part of the domain. On its whole domain, the utility may not be concave or differentiable. We establish the definition of the PHARA utility in Section \ref{sec:PHARA}. In Examples \ref{ex:M}-\ref{ex:HK} in Section \ref{sec:Examples}, Problems \eqref{prob:merton_basic}-\eqref{prob:constraint} can be represented by a PHARA utility maximization problem \eqref{prob-main}, up to some necessary transformations (e.g., converting to a unconstrained problem); the corresponding contexts and motivations are further demonstrated.




Throughout, we use the standard Black-Scholes model in a complete market, which is introduced in Section \ref{sec:model}. We provide a comprehensive overview on the model setting in Appendix \ref{app:correlated}. 
Our contributions are fourfold. 

First, we provide a unified formula of the optimal portfolio choice for the general PHARA utility (Section \ref{sec:formula}). The portfolio is expressed in a feedback form (Theorem \ref{thm_main}), i.e., an explicit function of the wealth level, given a Lagrangian multiplier ($y^*$ later). 
This unified formula includes optimal portfolios in the preceding literature as typical examples. As a direct formula, our result can be applied to quickly obtain the optimal portfolio in  utility maximization. We show the usage of our formula, how to specify parameters of the PHARA utility and how to directly write down the optimal portfolio, at Examples \ref{ex:M}-\ref{ex:PHARA_complex} in Section \ref{sec:revisit}. 





Second, we illustrate clear economic meanings for the formula (Section \ref{sec:econ}). The unified formula enjoys an advantage of analytical tractability and consists of four terms: \textbf{Merton relative risk aversion term}, \textbf{risk seeking term} (due to non-concavity), \textbf{loss aversion term} (due to benchmark levels) and \textbf{first-order risk aversion term} (due to non-differentiability). 
The division of the closed-form portfolio illustrates the risky behavior of the decision maker. Further, we propose an asymptotic analysis to our portfolio and the corresponding wealth process (Theorem \ref{thm_option table}). It is a generalization of the asymptotic approach in \cite{LL2023}. The main technique is asymptotic analysis and we use a more detailed analysis to the limiting of various terms in Theorem \ref{thm_option table}. One can directly apply Theorem \ref{thm_option table} for asymptotic financial analysis given a specific PHARA utility. As a summary, the optimal portfolio has a pattern of ``multiple-peak-multiple-valley" with the tail trend to the famous \cite{M1969}'s constant risky percentage. We find that not only non-concavity in the utility function causes great risk taking, but also non-differentiability greatly reduces risk taking. The former finding of non-concavity is studied in \cite{C2000} and \cite{KZ2007}, while the latter finding of non-differentiability is new.

Third, we compare the portfolios of the PHARA utility, the HARA utility and the SAHARA utility (symmetric asymptotic HARA; proposed by \cite{CPV2011}) analytically and numerically in Sections \ref{sec:comparison}-\ref{sec:nume_comparison}. We find that all these three corresponding portfolios approach to the Merton percentage when the market state is good. The HARA and PHARA portfolios become a pure risk-free investment when the market state is bad, while the SAHARA portfolio gambles more in a recession state. Different from the HARA portfolio's simple increasing structure with respect to the wealth, the PHARA portfolio has a relatively sophisticated structure: gamble (peak) at the non-concavity and become conservative (valley) at the non-differentiability of the utility.  

Fourth, we adopt the proposed PHARA utility and implement the portfolio using financial data for an empirical study (Section \ref{sec:Real-data study}). We set up our study by estimating our strategy parameters in an eight-year in-sample period and testing our strategy performance in a two-year out-of-sample period. We run a total of 1000 simulations in the out-of-sample period. We find that the proposed PHARA portfolio performs with high returns but high volatility, demonstrating positive alpha and a positive Sharpe Ratio. The majority of the calculated Sharpe Ratio values fall between 0 and 1.0, demonstrating that one of the drawbacks of the PHARA portfolio is high volatility. On the other hand, our simple returns are clustered around 100\% and the mean alpha value is greater than 0. This shows us that the PHARA portfolio can provide an advantage to the market.

In terms of methodology in deriving the optimal portfolio (Theorems \ref{thm_main}-\ref{thm_opt_strategy}), the basic approach is the martingale and duality method and the concavification technique. In a word, by duality arguments and concavification techniques, one can get the optimal terminal wealth, and hence obtain get the optimal wealth process and the optimal portfolio by the martingale representation theorem and It\^{o}'s formula. The martingale and duality method dates back to \cite{P1986}, \cite{KLS1987} and \cite{CH1989}. The concavification technique is established in \cite{C2000} and developed in \cite{KZ2007}, \cite{BS2014}, \cite{LL2020} and so on. This combined approach is frequently adopted in the literature, including \cite{HK2018}, \cite{HLLM2020}, \cite{DZ2020} and other papers above. Throughout, we implement this approach multiple times. We emphasize that our main theoretical novelty is the generality of the PHARA utility family, the delicate structure of the unified formula and the technical details of general asymptotic analysis. For a reference, we provide a basic theory on the concavification (also known as the concave envelope) in Appendix \ref{app:envelope}. 
Finally, Section \ref{sec:conclusion} concludes. All proofs are relegated to Appendix \ref{app:proof}.

\section{Model setting}\label{sec:model}

\subsection{Market model}
In the modern theory of investment, the Black-Scholes model is usually adopted to describe the financial market. Basically, there are a risk-free asset and multiple risky assets in the setting. A risk-free asset (e.g., bond) has a deterministic return rate and a zero volatility, while a risky asset (e.g., stock) has a higher expected return rate and a positive volatility. The key assumptions include that the price process of the risky asset is characterized by a geometric Brownian motion; see Appendix \ref{app:correlated} for a comprehensive overview. 

In this paper, we consider a standard Black-Scholes model in a complete market, consisting of one risk-free asset (i.e., a bond) and $m$ risky assets (i.e., stocks) where $m \geq 1$ is an integer. We introduce the notation as follows. Throughout, the bold symbol means a vector, while the plain symbol means a scalar. Fix the terminal time $T \in (0, \infty)$. The probability space $(\Omega, \mathcal{F}_T, \{\F_t\}_{0 \leq t \leq T}, \mathbb{P})$ describes the financial market. The equipped filtration $ \F = \{\F_t\}_{0 \leq t \leq T}$ is the one generated by an $m$-dimensional standard Brownian motion $\{\mathbf{W}_{t}\}_{0 \leq t \leq T} = \{(W_{1,t}, \cdots, W_{m,t})^\top\}_{0 \leq t \leq T}$ and further augmented by all the $\p$-null sets. 
For the risk-free asset, the interest rate is a constant $r>0$, and thus the price process $\{S_{0,t}\}_{0 \leq t \leq T}$ satisfies
\begin{equation}\label{eq:risk-free}
	\d S_{0,t} = r S_{0,t} \d t,  \;0 \leq t \leq T.
\end{equation}
For $m$ risky assets, we denote a vector of expected return rates by  $\boldsymbol{\mu} = (\mu_1, \cdots, \mu_m)^\top$ and an $m \times m$ matrix of volatility by $\boldsymbol{\sigma} = (\sigma_{i,j})_{1 \leq i, j \leq m}$. For $i = 1,\cdots,m$, the expected return rate of the $i$-th risky asset is $\mu_i>r$. The price process of the $i$-th risky asset $\{S_{i, t}\}_{0 \leq t \leq T}$ is a geometric Brownian motion satisfying the following stochastic differential equation:
\begin{equation}\label{eq:stock}
	\d S_{i, t} = S_{i, t}\left(\mu_i \d t + \sum_{j=1}^m\sigma_{i,j} \d W_{j,t}\right), \;0\leq t< T.
\end{equation}
We assume that there exists some $\epsilon > 0$ such that 
\begin{equation}\label{eq:ass_posi_defi}
	\boldsymbol{\eta}^\top (\boldsymbol{\sigma} \boldsymbol{\sigma}^\top) \boldsymbol{\eta} \geq \epsilon ||\boldsymbol{\eta}||_2^2 \;\; \text{ for any } \boldsymbol{\eta} \in \R^m,
\end{equation}
where $||\boldsymbol{\eta}||_2 := \left(\sum_{i=1}^m \eta_i^2\right)^{\frac{1}{2}}$. Under the assumption \eqref{eq:ass_posi_defi}, the matrices $\boldsymbol{\sigma}^\top$ and $\boldsymbol{\sigma}$ are invertible; see Lemma 2.1 of \cite{KLS1987}. We denote an $m$-dimensional vector by $\mathbf{1}_m = (1, \cdots, 1)^\top$. We further define a vector $\boldsymbol{\theta} := \boldsymbol{\sigma}^{-1} (\boldsymbol{\mu} - r \mathbf{1}_m)$, which actually means the market price of risk. Finally, we define the pricing kernel process $\{\xi_t\}_{0 \leq t \leq T}$ as follows:
\begin{equation}\label{eq:xi}
	\xi_t := \exp\left\{-\(r+  \frac{1}{2}||\boldsymbol{\theta}||_2^2\)t - \boldsymbol{\theta}^\top \mathbf{W}_t\right\}, \;0\leq t \leq T.
\end{equation}
The pricing kernel process is an important quantity in portfolio selection, and it is also known as the stochastic discount factor. The terminal pricing kernel variable $\xi_T$ plays a direct role in solving the terminal utility optimization problem.  

For any $i =1, \cdots, m$, for any time $t \in [0,T]$, we denote by $\pi_{i,t}$ the dollar amount invested in the $i$-th risky asset and denote by $\boldsymbol{\pi}_t = (\pi_{1,t}, \cdots, \pi_{m,t})^\top$ the $m$-dimensional portfolio vector. The process of the asset value $\{X_t\}_{0\leq t \leq T}$ is uniquely determined by the $m$-dimensional portfolio vector process $\{\boldsymbol{\pi}_t\}_{0 \leq t \leq T}$ and an initial value $x_0 \in \R$:
\begin{equation}\label{sde of x}
	\begin{aligned}
		\d X_t &= \sum_{i=1}^{m} \pi_{i,t} \left( \mu_i \d t + \sum_{j=1}^m \sigma_{i,j} \d W_{j,t} \right) + \left(X_t-\sum_{i=1}^{m} \pi_{i, t}\right) r \d t\\
		&= \left( r X_t + \boldsymbol{\pi}_t^\top (\boldsymbol{\mu} - r \mathbf{1}_m)\right) \d t + \boldsymbol{\pi}_t^\top \boldsymbol{\sigma} \d \mathbf{W}_t , \;\;\;\; X_0 = x_0.
	\end{aligned}
\end{equation}
A portfolio $\boldsymbol{\pi} = \{\boldsymbol{\pi}_t\}_{0 \leq t \leq T}$ is called admissible if and only if $\boldsymbol{\pi}$ is an $\{\F_t\}_{0 \leq t \leq T}$-progressively measurable $\R^m$-valued process and  $\int_0^T ||\boldsymbol{\pi}_{t}||_2^2 \d t < \infty$ almost surely. Hence $\boldsymbol{\pi}$ is allowed to do short-selling, i.e., it may happen that $\pi_{i, t} < 0$ for some $i = 1, \cdots, m$ and $t \in [0, T]$.  
The decision maker solves an expected utility optimization problem to conduct portfolio selection: 
\begin{equation}\label{prob-main}
	\begin{aligned}
		& \max_{\boldsymbol{\pi} \in \mathcal{V}} \E[ U(X_T)],
	\end{aligned}
\end{equation}
where $U$ is the utility function and $\mathcal{V}$ is the collection of all the admissible portfolios. 

\subsection{PHARA utility}\label{sec:PHARA}
In our context, the utility function $U$ may not be concave and cause difficulties in the solving procedure. A highly related concept is the concave envelope. Let $\D \subseteq \R$ be a convex set. Denote a continuous function by $U: \D \to \R$, where the domain of $U$ is denoted by $\text{dom } U = \D$. The concave envelope of $U$ (denoted by $U^{**}$) is defined as the smallest 
continuous concave function larger than $U$. That is, for $x \in \D$, 
\begin{equation}\label{eq:conv_dominating}
	U^{**}(x) := \inf\{h(x): h \text{ maps $\D$ to $\R$, $h$ is a concave and continuous function on $\D$ and } h \geq U\}.
\end{equation}
Hence, for a continuous utility function $U$, the concave envelope $U^{**}$ is concave, 
continuous and larger than $U$, and $U^{**}$ is the smallest function satisfying these three conditions; see later Figure \ref{fig} as an example for graphical illustration.
We introduce the basic theory on the non-concave function and its concave envelope in Appendix \ref{app:envelope}, where we provide a procedure of generating the concave envelope from the original function.

The concave envelope plays an important role in portfolio selection problems with non-concave utilities; see Appendix \ref{app:envelope}. Indeed, if $\xi_T$ in \eqref{eq:xi} has a continuous distribution, Problem \eqref{prob-main} has the same optimal solution as the following problem with the utility $U$ replaced by the concave envelope $U^{**}$:
\begin{equation}\label{prob:concavification}
	\max_{\boldsymbol{\pi} \in \mathcal{V}} \E\left[ U^{**}(X_T) \right].
\end{equation}
It is shown in Theorem 2.5 in \cite{BS2014}, Theorem 3.3 in \cite{LL2020} and also the proof of Proposition \ref{Thm_opt_wealth_process} later. 
Thus, we can study $U^{**}$ instead of $U$ and later obtain the same optimal portfolio for both Problems \eqref{prob-main} and \eqref{prob:concavification}.

Next, we define the main object in this paper, piecewise HARA (PHARA) utility. Intuitively, it means being a HARA function on each part of its domain. 

\begin{definition}[PHARA utility]
	\label{def_PHARA}
	Define the function
	$\tilde{U}: \R \to \R$ (with relative risk aversion parameter $R \in [0, \infty]$, benchmark level $A \in \R$, partition point $\hat{x} \in \R$, utility value $u \in \R$, slope $\gamma \in (0, \infty)$, absolute risk aversion parameter $\alpha \in (0, \infty)$) as
	\begin{equation}%
		\label{eq_def_U_tilde}
		\begin{aligned}
			\tilde{U}(x; R, A, \hat{x}, u, \gamma, \alpha) 
			:=
			\left\{
			\begin{aligned}
				& \gamma ( x-\hat{x} ) + u, && \text{ if } R = 0;\\
				& \gamma (\hat{x}-A) \log\left(\frac{x-A}{\hat{x}-A}\right) + u, && \text{ if } R = 1, \hat{x} \neq A;\\
				& \gamma \frac{\hat{x}-A}{1-R} \left( \left(\frac{x-A}{\hat{x}-A}\right)^{1-R}-1 \right) + u, && \text{ if } R \in (0, 1) \cup (1, \infty), \hat{x} \neq A;\\
				& 
				-\frac{\gamma}{\alpha} \left( e^{-\alpha(x-\hat{x})} -1 \right) +u, && \text{ if } R = \infty, A = -\infty, \alpha \in (0, \infty). 
			\end{aligned}
			\right.
		\end{aligned}
	\end{equation}
	Set the domain $\mathcal{D} := [a_0, \infty)$ or $(a_0, \infty)$, where $a_0  := \inf\{a \in \R | U(a) > -\infty\}$. 
	A function $U: \mathcal{D} \to \R$ is a \textbf{piecewise HARA utility} if and only if these exist a partition 
	$\{a_k\}_{k = 0}^{n+1}$ and a family of parameter pairs $\{(R_k, A_k, \alpha_k)\}_{k=0}^n$ (note that it is no need to specify $A_k$ for $k$ satisfying $R_k = 0$, i.e., on the linear part) such that 
	\begin{enumerate}[(i)]
		\item $n \geq 0$, $a_0 < a_1 < \cdots < a_{n+1}$, $a_1, \cdots, a_n \in \R$, $a_{n+1} = \infty$;
		\item $U$ is increasing and continuous on $\mathcal{D}$;
		\item 
		\begin{enumerate}[(a)]
			\item
			If $n = 0$ and $\mathcal{D} = (a_0, \infty)$, or $n =0$ and $\mathcal{D} = [a_0, \infty)$ and $A_0 = a_0$, then 
			$U(x) = \tilde{U}^0(x; R, A, u, \gamma, \alpha)$ for any $x \in (a_0, \infty)$, where 
			\begin{equation}%
				\label{eq_def_U_tilde_0}
				\begin{aligned}
					\tilde{U}^0(x; R, A, u, \gamma, \alpha) 
					:=
					\left\{
					\begin{aligned}
						& \gamma x + u, && \text{ if } R = 0;\\
						& \gamma \log (x-A) + u, && \text{ if } R = 1;\\
						& \frac{\gamma}{1-R}(x-A)^{1-R} + u, && \text{ if } R \in (0, 1) \cup (1, \infty);\\
						& 
						-\frac{\gamma}{\alpha} \left( e^{-\alpha x} -1 \right) +u, && \text{ if } R = \infty, A = -\infty, \alpha \in (0, \infty). 
					\end{aligned}
					\right.
				\end{aligned}
			\end{equation}
			\item 
			If $n = 0$ and $\mathcal{D} = [a_0, \infty)$, then $U(x) = \tilde{U}(x; R_0, A_0, a_0, U(a_0), U'(a_0+), \alpha_0)$ for any $x \in [a_0, \infty)$.
			\item 
			If $n \geq 1$, for $k = 0$, $U(x) = \tilde{U}(x; R_0, A_0, a_1, U(a_1), U'(a_1-), \alpha_0)$ for any $x \in (a_0, a_1)$; 
			for any $k\in\{1, 2, \cdots, n\}$, $U(x) = \tilde{U}(x; R_k, A_k, a_k, U(a_k), U'(a_k+), \alpha_k)$ for any $x \in (a_k, a_{k+1})$.
		\end{enumerate}
		\item For $n$, $R_n \neq 0$ (i.e., $U$ is not linear on the last interval $(a_n, \infty)$). 
	\end{enumerate}
\end{definition}
\begin{remark}
	Technically, Definition \ref{def_PHARA} is classified in categories (a)-(c) because of the different settings under $n = 0$ or $n  \geq 1$ and $\D = [a_0, \infty)$ or $\D = (a_0, \infty)$. In (b)-(c), one could use information of partition points $a_0$, $a_1$, etc., to describe the expression of the utility. In (a), there is no specific partition point with information of utility values and slopes, so one should specify the parameters directly.  
\end{remark}

We show in the following proposition that the concave envelope preserves the property of HARA. 
\begin{proposition}\label{prop:PHARA}
	If $U$ is a PHARA utility function, then $U^{**}$ is also a PHARA utility function. Conversely, if $U^{**}$ is PHARA, $U$ may not be HARA. 
\end{proposition}
Hence, the PHARA utility family has a high level of generality and includes many explicit utility functions, such as (piecewise) power/log/exponential and S-shaped ones. In fact, our result only requires that the concave envelope $U^{**}$ is PHARA. 
This increases the generality of the choice of utility functions, because: (i) $U^{**}$ is PHARA if $U$ is PHARA;  
(ii) $U^{**}$ is possible to be PHARA even if $U$ is not PHARA; (iii) PHARA utilities can be used as building blocks to approximate/characterize an unknown preference. In many cases, the explicit expression of a utility function is unknown, while only information on non-concavity and differentiablity of some intervals is available; see, e.g., \cite{LL2020}. The PHARA utilities can hence serve as basic elements for approximation. Finally, for simplicity, we use the notation: $\gamma_k^+ = U'(a_k+)$, $\gamma_k^- = U'(a_k-)$, where $\gamma_0^- = \infty$ and $\gamma_{n+1}^- = 0$.

\begin{remark}
	There exist some utility functions such that $U(a_0) = -\infty$; e.g., $U(x) = 3 \log (x)$ or $U(x) = -\frac{1}{2} x^{-2}$, $x \in (a_0, a_1)$, where $a_0 = 0$ and $a_1 \in (0, \infty)$. These cases are also included in the PHARA utility family. For these cases, it is only possible to appear on the first interval of the domain ($k = 0$) as $U(a_0) = -\infty$. According to Definition \ref{def_PHARA}, we can write 
	$$
	\begin{aligned}
		&U(x) = \tilde{U}\left(x; R_0 = 1, A_0 = a_0=0, a_1, U(a_1) = 3 \log(a_1), U'(a_1) = 3 a_1^{-1}\right), && \text{ for } U(x) = 3 \log(x), \; x \in (a_0, a_1);\\
		&U(x) = \tilde{U}\left(x; R_0 = 3, A_0 = a_0=0, a_1, U(a_1) = -\frac{1}{2} a_1^{-2}, U'(a_1) = a_1^{-3}\right), && \text{ for } U(x) = -\frac{1}{2} x^{-2}, \; x \in (a_0, a_1).
	\end{aligned}
	$$ 
\end{remark}

\subsection{Examples}\label{sec:Examples}
In this section, we show that plenty of models in the literature, either classic or non-standard, become examples of our PHARA utility in Definition \ref{def_PHARA}. Later in Section \ref{sec:revisit} we establish the parameter tables in terms of a PHARA utility and use a unified formula to write down the optimal portfolio of these examples. 
\begin{example}\label{ex:M}
	In the context of \cite{M1969}, the decision maker (i.e., a single investor) selects the optimal portfolio using the above standard Black-Scholes model with $m$ risky assets and one risk-free asset. As discussed above, the initial wealth for investment is $x_0 > 0$ and the investment period is $T > 0$. We use the same notation in the model setting \eqref{eq:risk-free}-\eqref{prob-main}.
	\begin{itemize}
		\item
		The decision maker has the following CRRA utility with $R > 0$:
		\begin{equation}\label{eq:U:CRRA}
			U(x) = \left\{
			\begin{aligned}
				& \frac{1}{1-R} x^{1-R}, && x \in [0, \infty), &&& \text{ if } R \neq 1;\\
				& \log(x), && x \in (0, \infty), &&& \text{ if } R = 1,
			\end{aligned}
			\right.
		\end{equation}
		Hence, $U$ is PHARA and there is only one part in the domain ($n+1 = 1$); 
		see parameters in Table \ref{tab_CRRA}. In addition, implied by this utility, the condition of nonnegative wealth holds:
		\begin{equation}\label{eq:nonnegative}
			X_t \geq 0, \; \text{ for any } t \in [0, T).     
		\end{equation} 
		
		\item
		The decision maker has the following CRRA utility with $\alpha > 0$:
		\begin{equation}\label{eq:U:CARA}
			U(x) = -\frac{1}{\alpha} e^{- \alpha x}, \; x \in (-\infty, \infty).
		\end{equation}
		Hence, $U$ is PHARA and there is only one part in the domain ($n+1 = 1$); 
		see parameters in Table \ref{tab_CARA}. As the domain is $(-\infty, \infty)$, the wealth at $t$ can be negative and the condition of nonnegative wealth \eqref{eq:nonnegative} need not hold. 
	\end{itemize}
\end{example}

\begin{example}\label{ex:C}
	In the context of \cite{C2000}, a hedge fund manager makes the investment decision. The financial market also follows a multi-asset Black-Scholes model \eqref{eq:risk-free}-\eqref{prob-main} and the condition of nonnegative wealth:
	\begin{equation}
		X_t \geq 0, \; \text{ for any } t \in [0, T).     
	\end{equation} 
	According to the notation of \cite{C2000}, there is an option compensation scheme under which the decision maker (i.e., the manager) has the following payoff function:
	\begin{equation}
		\Theta(x) = \alpha (x - B_T)^{+} + K, \; x \in [0, \infty),
	\end{equation}
	where $B_T := B_0 e^{r T}$ is a riskless benchmark, $B_0$ is a constant, $\alpha > 0$ is the number of options and $K > 0$ is a constant amount of wealth. The option compensation scheme means that the decision maker has his/her own payoff $\Theta(X_T)$ if the terminal fund wealth is $X_T$, where the function $\Theta$ is actually an option.  
	The decision maker has a HARA utility $\widehat{U}(W)$ on his/her own payoff $W$:
	\begin{equation}
		\widehat{U}(W) = \frac{1-\gamma}{\gamma} \left( \frac{A(W-\omega)}{1-\gamma} \right)^\gamma, \; W \in [K, \infty),
	\end{equation}
	where parameters include $\gamma < 1$ (as above, $\widehat{U}$ is $\log$ if $\gamma = 0$), $\omega < K$ and $A > 0$.

	Thus, the decision maker's actual utility is defined by a composition function on the fund wealth $x \geq 0$: 
	\begin{equation}
		U(x) := \widehat{U} \circ \Theta(x)=\left\{
		\begin{aligned}
			&\frac{(1-\gamma)^{1-\gamma}}{\gamma} A^\gamma (K-\omega)^\gamma, && \text{ if } x \in [0, B_T);\\
			&\frac{(1-\gamma)^{1-\gamma}}{\gamma} A^\gamma \alpha^\gamma \left( x - \left(B_T - \frac{K-\omega}{\alpha}\right) \right)^\gamma, &&  \text{ if } x \in [B_T, \infty).
		\end{aligned}
		\right.
	\end{equation}
	The concave envelope of $U$ is denoted by $U^{**}$ and given by
	\begin{equation}\label{eq:C_concavification}
		U^{**}(x)=\left\{
		\begin{aligned}
			& U(0) + U'(\hat{x}) x, &&  \text{ if } x \in [0, \hat{x});\\
			& U(x), &&  \text{ if } x \in [\hat{x}, \infty),
		\end{aligned}
		\right.
	\end{equation}
	where 
	$\hat{x} > 0$ is the solution of
	\begin{eqnarray}
		\frac{U(\hat{x}) - U(0)}{\hat{x}} = U'(\hat{x}).
	\end{eqnarray}
	Hence, both $U$ and $U^{**}$ are PHARA; see parameters of $U^{**}$ in Table \ref{tab_CARA}.
\end{example}

\begin{example}\label{ex:LSW}
	In the context of \cite{LSW2017},  
	the decision maker is an insurer with participating insurance contracts and aims to solve the optimal asset allocation. The market setting is a standard one-dimensional Black-Scholes model \eqref{eq:risk-free}-\eqref{prob-main} with $m=1$. 
	In this example, the decision maker has a behavioral preference (i.e., S-shaped preference):    
	\begin{equation}\label{eq:S-shaped:LSW}
		\widehat{U}(W)=\left\{
		\begin{aligned}
			&W^\gamma,\;\; &&  \text{ if } W \geq 0;\\
			&-\lambda (-W)^\gamma,\;\; &&  \text{ if } W < 0.
		\end{aligned}
		\right.
	\end{equation}
	It means that the decision maker shows risk aversion when exceeding the reference point 0 and loss aversion when falling below. Moreover, there is a so-called participating contract such that according to different levels, the total wealth $X_T$ is shared between the insurer $\Psi(X_T)$ and the policyholder $X_T - \Psi(X_T)$ with different proportions. The contract hence provides a minimal guarantee for the policyholder and an high-reward incentive for the insurer. Mathematically, the contract results in a piecewise linear payoff for the insurer (in the following discussion, $\gamma, \delta, \alpha \in (0, 1)$ and $L_T^g > 0$ are all constants given in \cite{LSW2017}):
	\begin{equation}
		\Psi(X_T)=\left\{
		\begin{aligned}
			& (1-\delta \alpha)X_T - (1-\delta)\LT,\;\; &&  \text{ if } X_T \geq \frac{\LT}{\alpha};\\
			& X_T - \LT,\;\; &&  \text{ if } \LT \leq X_T < \frac{\LT}{\alpha};\\
			& 0, \;\; &&  \text{ if } 0 \leq X_T < \LT,
		\end{aligned}
		\right.
	\end{equation}
	while the policyholder's payoff is $X_T - \Psi(X_T)$. 
	Hence, the optimization problem of the insurer is to maximize the expected utility of the corresponding payoff, which is exactly Problem \eqref{prob-main} with
	the objective function defined as $U := \widehat{U} \circ \Psi$. Specifically, 
	\begin{equation}
		U(x) := \widehat{U}(\Psi(x))=\left\{
		\begin{aligned}
			& (1-\delta \alpha)^\gamma \left(x - \frac{1-\delta}{1-\delta \alpha}\LT\right)^\gamma,\;\; &&  \text{ if } x \geq \frac{\LT}{\alpha};\\
			& (x - \LT)^\gamma,\;\; &&  \text{ if } \LT \leq x < \frac{\LT}{\alpha};\\
			& 0, \;\; &&  \text{ if } 0 \leq x < \LT.
		\end{aligned}
		\right.
	\end{equation}
	It can be checked straightforward that $U$ is PHARA. 
	We define $U^{**}$ as the concave envelope of $U$. In the Case A1 ($1-\gamma > \alpha$), according to the notation of \cite{LSW2017}, we have
	\begin{equation}\label{eq:LSW_concavification}
		U^{**}(x)=\left\{
		\begin{aligned}
			& (1-\delta \alpha)^\gamma \left(x - \frac{1-\delta}{1-\delta \alpha}\LT \right)^\gamma,\;\; && \text{ if } x \geq \frac{\LT}{\alpha};\\
			& (x - \LT)^\gamma,\;\; &&  \text{ if } \frac{\LT}{1-\gamma} \leq x < \frac{\LT}{\alpha};\\
			& 0, \;\; &&  \text{ if } 0 \leq x < \frac{\LT}{1-\gamma},
		\end{aligned}
		\right.
	\end{equation}
	where $\frac{\LT}{1-\gamma}$ is a tangent point. Hence, both $U$ and $U^{**}$ is PHARA; see parameters of $U^{**}$ in Table \ref{tab2}. 
\end{example}

\begin{example}\label{ex:HK}
	In the context of \cite{HK2018}, a hedge fund manager serves as a decision maker. The financial market has a risky asset and a risk-free asset, which follows the Black-Scholes model \eqref{eq:risk-free}-\eqref{prob-main} with $m = 1$. 
	There is a first-loss incentive scheme under which the manager has the payoff:
	\begin{equation}
		\Theta(x)=\left\{
		\begin{aligned}
			&-\omega x_0, &&  \text{ if } x \in [b x_0, (1-\omega)x_0);\\
			&x - x_0, &&  \text{ if } x \in [(1-\omega)x_0, x_0];\\
			&\left( \omega + \alpha (1-\omega) \right) (x - x_0), &&  \text{ if } x \in [x_0, \infty).
		\end{aligned}
		\right.
	\end{equation}
	The decision maker has the payoff $\Theta(X_T)$ for his/her own if the terminal fund wealth is $X_T$.  Similar to \eqref{eq:S-shaped:LSW} in Example \ref{ex:LSW}, the decision maker is assumed to have an S-shaped preference function: 
	\begin{equation}
		\widehat{U}(W)=\left\{
		\begin{aligned}
			&W^p, &&   \text{ if } W \geq 0;\\
			&-\lambda (-W)^q, &&  \text{ if } W \leq 0.
		\end{aligned}
		\right.
	\end{equation}	
	In addition, there is a lower-bound liquidation constraint for the wealth:
	\begin{equation}
		X_T \geq b x_0, \; \text{ a.s.}.     
	\end{equation}  
	Here $b x_0 > 0$ means a threshold that the terminal fund wealth has to exceed in order to avoid liquidation. By multiplying the discount factor $e^{r (t -T)}$, one obtains a liquidation constraint $X_t \geq b x_0 e^{-r(T-t)}$ for any $t \in [0, T]$.    
	As a result, the decision maker's objective is defined by a composition function on the fund wealth $x \geq b x_0$: 
	\begin{equation}
		U(x) := \widehat{U} \circ \Theta(x) = \left\{
		\begin{aligned}
			&-\lambda (\omega x_0)^p, &&  \text{ if } x \in [b x_0, (1-\omega) x_0);\\
			&-\lambda (x_0 - x)^p, &&  \text{ if } x \in [(1-\omega) x_0, x_0);\\
			&\( \omega + \alpha (1-\omega) \)^p (x - x_0)^p, &&  \text{ if } x \in [x_0, \infty).
		\end{aligned}
		\right.
	\end{equation}
	and the concave envelope becomes
	\begin{equation}\label{eq:HK_concavification}
		U^{**}(x) := \left\{
		\begin{aligned}
			& U(a_2) + U'(a_2) (x - a_2), && x \in [b x_0, (1+c^*) x_0);\\
			&\( \omega + \alpha (1-\omega) \)^p (x - x_0)^p, && x \in [(1+c^*) x_0, \infty).
		\end{aligned}
		\right.
	\end{equation}
	where $(1+c^*) x_0 > x_0$ is the solution of
	\begin{equation}
		\frac{U((1+c^*) x_0) - U(b x_0)}{(1+c^*) x_0 - b x_0} = U'((1+c^*) x_0).
	\end{equation}
	Hence, both $U$ and $U^{**}$ are PHARA; see parameters of $U^{**}$ in Table \ref{tab:HK2018}.
\end{example}

\begin{example}\label{ex:PHARA_complex}
	We provide another example with a graph for the PHARA utility family to show its generality. Figure \ref{fig} shows the contour of a general (highly non-concave and non-differentiable) utility. The concave envelope is the smallest concave and continuous function dominating the original utility function, and hence involves a lot of tangent lines and linear connections; see Appendix \ref{app:envelope} for more details. Later we will show that, using the unified formula in Theorem \ref{thm_main}, we can directly compute the closed-form optimal portfolio for this example and know the risk-taking behavior on each part of the utility's domain. This example is revisited in the Section \ref{sec:numerical}. 
	
	\begin{figure}[H]
		\begin{minipage}{0.48\textwidth}
			\centering
			\includegraphics[width=\textwidth]{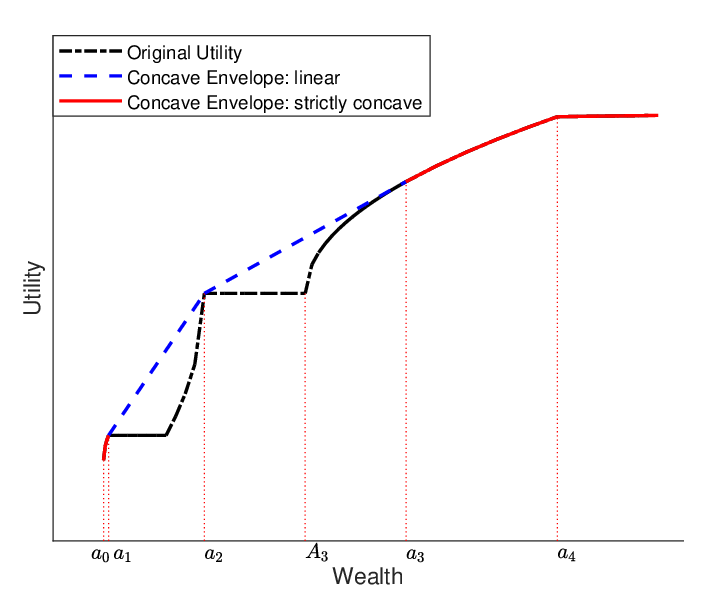}
			\label{0CPT}
		\end{minipage}
		\hspace{1em plus 1fill}
		\begin{minipage}{0.48\textwidth}
			\begin{equation*}
				U(x) =\left\{
				\begin{aligned}
					& k_1  (x - a_0)^{1 - R} - k_1 (a_1- a_0)^{1-R} \\
					& \quad \quad - \lambda (a_2 - a_{1,2})^{q}, \quad\; \text{ if } a_0 \leq  x < a_1;\\
					& -\lambda (a_2 - a_{1,2})^{q}, \quad\quad\quad\; \text{ if } a_1 \leq x  < a_{1,2};\\
					& -\lambda (a_2 - x )^{1-R}, \quad\quad\quad \text{ if } a_{1,2} \leq x < a_2;\\
					& 0, \quad\quad\quad\quad\quad\quad\quad\quad\quad\;\; \text{ if } a_2 \leq x < A_3;\\
					& k_1 (x - A_3)^{1-R}, \quad\quad\quad\;\; \text{ if } A_3 \leq x < a_4;\\
					&  k_2 (x - A_4)^{1-R}+ k_1 A_3^{1-R}  \\
					& \quad \quad - k_2 (a_4 - A_4)^{1-R}, \text{ if } x \geq a_4.
				\end{aligned}
				\right.
			\end{equation*}
		\end{minipage}
		\caption{An example of a PHARA utility function $U$ and its concave envelope $U^{**}$ ($U^{**}$ is also PHARA: $n = 4$, $A_0 = a_0 = 4$, $A_4 = A_3 = 20$, $a_5 = \infty$, non-differentiable points: $a_0 = 4, ~a_1 = 4.4, ~a_2 = 12, ~a_4 = 40$, tangent point: $a_3 = 28$. Other parameters: $a_{1,2} = 8.96$, $k_1 = 0.24^{0.5}$, $k_2 = 0.02$, $\lambda = 1.01$, $R = 0.5$, $q = 0.5$.). Note that if the concave envelope is strictly concave on some interval, the original utility function is the same as its concave envelope on this interval.}
		\label{fig}
	\end{figure}
\end{example}



%
%
%
%
%
%
%

\section{A unified formula of the optimal portfolio}
\label{sec:formula}

We proceed to show the optimal terminal wealth, the wealth process and the portfolio vector of PHARA utilities for Problem \eqref{prob-main}. The basic approach is the martingale method and the concavification technique. 
Above all, Proposition \ref{Thm_opt_wealth_process} proposes a five-term division of the optimal wealth process. The proof is included in Appendix \ref{app:proof}.


\begin{proposition}\label{Thm_opt_wealth_process}
	Suppose that the concave envelope $U^{**}$ has the form in Definition \ref{def_PHARA} and $R_k \in [0, \infty]$ for each $k \in \{0,1,\cdots, n\}$. For Problem \eqref{prob-main},
	
	\noindent
	(1) the optimal terminal wealth is given by
	\begin{equation}
		\label{eq_terminal}
		\begin{aligned}
			X_T^* =  \sum_{k=0}^n \biggr\{&a_k \id_{\left\{y^* \xi_T \in \left(\gamma_k^{+}, \gamma_k^{-}\right)\right\}}\\ 
			+&  \left( A_k + \left(\frac{\gamma_k^{+}}{y^* \xi_T }\right)^{\!R_k^{-1}}\!\!(a_k - A_k) \right) \id_{\left\{y^* \xi_T \in \left(\gamma_{k+1}^{-}, \gamma_k^{+}\right)\right\}} \id_{\{ R_k \neq 0,\infty \}}\\ 
			+& \left(a_k + \frac{1}{\alpha_k} \log\left(\frac{\gamma_k^{+}}{y^* \xi_T}\right)\right) \id_{\left\{y^* \xi_T \in \left(\gamma_{k+1}^{-}, \gamma_k^{+}\right)\right\}} \id_{\left\{ R_k = \infty, A_k = -\infty, \alpha_k > 0 \right\}} \biggr\}, \text{ a.s., }
		\end{aligned}
	\end{equation}
	where $y^*$ is a unique positive number that satisfies 	\begin{equation}\label{equation of nu*}
		\begin{aligned}
			\E[\xi_T X_T^*]  = x_0.
		\end{aligned}
	\end{equation}
	
	\noindent
	(2) The optimal wealth at time $t \in [0,T)$ is given by 
	\begin{equation}\label{eq:opt_process}
		\begin{aligned}
			X_t^* &:= X^D_{t} + X^A_{t} + X^{\bar{A}}_{t} + X^R_{t} + X^{\bar{R}}_{t}\\
			&= \sum_{k=0}^n \left( X^D_{t,k} + X^A_{t,k} + X^{\bar{A}}_{t,k} + X^R_{t,k} + X^{\bar{R}}_{t,k} \right),
		\end{aligned}
	\end{equation}
	where the terms are given by
	\begin{equation}
		\label{eq:wealth_five_term}
		\begin{aligned}
			&X^D_{t,k} := e^{-r(T-t)} a_k  \left[ \Phi\left(d_1\left(\frac{\gamma^+_k}{y^*\xi_t}\right)\right) -\Phi\left(d_1\left(\frac{\gamma^-_k}{y^* \xi_t}\right)\right) \right], \\
			&X^A_{t,k} :=  e^{-r(T-t)} A_k \left[ \Phi\left(d_1\left(\frac{\gamma^-_{k+1}}{y^*\xi_t}\right)\right) -\Phi\left(d_1\left(\frac{\gamma^+_k}{y^*\xi_t}\right)\right) \right] \id_{\left\{ R_k \neq 0,\infty \right\}},\\
			&X^{\bar{A}}_{t,k} := e^{-r(T-t)} \left[ a_k+\frac{-||\boldsymbol{\theta}||_2 \sqrt{T-t}}{\alpha_k} d_1\left(\frac{\gamma_k^{+}}{y^* \xi_t}\right)\right]\left[\Phi \left(d_1\left(\frac{\gamma_{k+1}^{-}}{y^*\xi_t}\right)\right)-\Phi\left(d_1\left(\frac{\gamma_k^{+}}{y^*\xi_t}\right)\right)\right] \id_{ \left\{ R_k = \infty, A_k = -\infty, \alpha_k > 0 \right\} }, \\ 
			&X^R_{t,k} :=  e^{-r(T-t)} (a_k-A_k) \frac{\Phi'\left(d_1\left(\frac{\gamma^+_k}{y^*\xi_t}\right)\right)}{\Phi'\left(d^k\left(\frac{\gamma^+_k}{y^*\xi_t}\right)\right)}\left[ \Phi\left(d^k\left(\frac{\gamma^-_{k+1}}{y^*\xi_t}\right)\right) -\Phi\left(d^k\left(\frac{\gamma^+_k}{y^*\xi_t}\right)\right) \right] \id_{\left\{ R_k \neq 0,\infty \right\}},\\
			&X^{\bar{R}}_{t,k} := e^{-r(T-t)} \frac{-||\boldsymbol{\theta}||_2 \sqrt{T-t}}{\alpha_k} \left[\Phi' \left(d_1\left(\frac{\gamma_{k+1}^{-}}{y^*\xi_t}\right)\right)-\Phi'\left(d_1\left(\frac{\gamma_k^{+}}{y^*\xi_t}\right)\right)\right] \id_{ \left\{ R_k = \infty, A_k = -\infty, \alpha_k > 0  \right\} }, 
		\end{aligned}
	\end{equation}
	and a transformation function is denoted by $d(z, h) := \frac{1}{-||\boldsymbol{\theta}||_2 \sqrt{T-t}}\left(\log(z) +\left(r+\frac{||\boldsymbol{\theta}||_2^2}{2}\right)(T-t)\right)+ h ||\boldsymbol{\theta}||_2 \sqrt{T-t}$. We further set $d_0(z)  := d(z, 0),~d_1(z) := d(z, 1), ~d^{k}(z) := d(z, 1 - R_k^{-1}), ~k= 0, \cdots, n.$
\end{proposition}

The optimal terminal wealth \eqref{eq_terminal} contains three terms, while the optimal wealth process \eqref{eq:opt_process} contains five terms. The term $X^D_{t, k}$ comes from the non-differentiable point $a_k$ of the utility function (the left- and right- derivatives are not equal: $\gamma_{k}^{+} \neq \gamma_{k}^{-}$). We call $X^D_{t}$ the \textbf{first-order risk aversion term}, and later in Theorem \ref{thm_main} it leads to the term $\boldsymbol{\pi}_t^{(4)}$ with a same name in the optimal portfolio \eqref{eq:main}. The name is because compared to \cite{P1964}'s classic ``second-order" risk aversion caused by concavity, the non-differentiability causes ``first-order" risk aversion; see \cite{SS1997}.  
We call the terms $X_{t}^{A}$ and $X_{t}^{\bar{A}}$ the \textbf{loss aversion terms}, as they originate from a loss-averse benchmark (e.g., $A_k$). They lead to the term $\boldsymbol{\pi}_t^{(3)}$ with a same name in the portfolio. We call the terms $X_{t}^{R}$ and $X_{t}^{\bar{R}}$ the \textbf{risk seeking terms} and they lead to the term $\boldsymbol{\pi}_t^{(2)}$ with a same name in the portfolio. 
These names are explained  in detailed after Theorem \ref{thm_main}.  
Lastly, the terms $X_{t, k}^{A}$ and $X_{t, k}^{R}$ come from a CRRA utility on the part $(a_k, a_{k+1})$, while $X_{t, k}^{\bar{A}}$ and $X_{t, k}^{\bar{R}}$ come from a CARA utility on the part.

Now we present the unified formula of the optimal portfolio of PHARA utilities for Problem \eqref{prob-main}. The proof is included in Appendix \ref{app:proof}. Basically, we apply the martingale method on the optimal wealth $X_t^*$ of Proposition \ref{Thm_opt_wealth_process} and hence obtain the optimal portfolio vector. 
\begin{theorem}\label{thm_main}
	Suppose that the concave envelope $U^{**}$ has the form in Definition \ref{def_PHARA} 
	and $R_k \in \{ R, 0 \}$ for each $k \in \{0,1,\cdots, n\}$ (where $R \in (0, \infty)$). For Problem \eqref{prob-main}, the optimal portfolio vector at time $t \in [0, T)$ is
	\begin{equation}\label{eq:main}
		\begin{aligned}
			\boldsymbol{\pi}^*_t &= \underbrace{\frac{(\boldsymbol{\sigma}^\top)^{-1} \boldsymbol{\theta}}{R} X^*_t}_{\text{Merton term}} + \underbrace{\frac{e^{-r(T-t)}}{\sqrt{T-t}} \frac{(\boldsymbol{\sigma}^\top)^{-1} \boldsymbol{\theta}}{||\boldsymbol{\theta}||_2} \sum_{k=0}^{n} (a_{k+1} - a_k) \Phi'\(d_1\(\frac{\gamma_k^{+}}{y^* \xi_t}\)\) \id_{\{R_k = 0 \}}}_{\text{risk seeking}}\\
			&\quad \quad \quad \underbrace{-e^{-r(T-t)}\frac{(\boldsymbol{\sigma}^\top)^{-1} \boldsymbol{\theta}}{R} \sum_{k=0}^{n} A_k q_k \id_{\{R_k \neq 0 \}}}_{\text{loss aversion}}  \underbrace{-e^{-r(T-t)}\frac{(\boldsymbol{\sigma}^\top)^{-1} \boldsymbol{\theta}}{R} \sum_{k=0}^{n} a_k p_k }_{\text{first-order risk aversion}}\\
			&:= \boldsymbol{\pi}_t^{(1)} + \boldsymbol{\pi}_t^{(2)} + \boldsymbol{\pi}_t^{(3)} + \boldsymbol{\pi}_t^{(4)},
		\end{aligned}
	\end{equation}
	where $X_t^*$ is the optimal wealth amount at time $t$, $y^*$ is a unique positive number that satisfies \eqref{equation of nu*},  
	$\Phi(\cdot)$ is the standard normal cumulative distribution function and we define
	\begin{equation}\label{eq:d1}
		d_1(z) := \frac{1}{-||\boldsymbol{\theta}||_2 \sqrt{T-t}}\left(\log(z) +\left(r-\frac{||\boldsymbol{\theta}||_2^2}{2}\right)(T-t)\right), \;\; z > 0,
	\end{equation}		
	\begin{equation}\label{eq:pkqk}
		p_k := \Phi\left(d_1\left(\frac{\gamma_k^{+}}{y^* \xi_t}\right)\right) - \Phi\left(d_1\left(\frac{\gamma_k^{-}}{y^* \xi_t}\right)\right),\;\; q_k :=\Phi\left(d_1\left(\frac{\gamma_{k+1}^{-}}{y^* \xi_t}\right)\right) - \Phi\left(d_1\left(\frac{\gamma_{k}^{+}}{y^* \xi_t}\right)\right), \;\; k = 0, 1, \cdots, n. 
	\end{equation}
\end{theorem}

\begin{remark}
	\begin{enumerate}[(i)]
		\item As $\sum_{k=0}^n p_k + \sum_{k=0}^n q_k = 1$, $p_k$ and $q_k$ are interpreted to be probabilities.
		\item By definition and computation, we have $q_k = 0$ if $R_k = 0$; in this case, there is no need to specify $A_k$. 
		\item We have $p_k = 0$ if $a_k$ is differentiable. 
	\end{enumerate}
\end{remark}

Under a general PHARA utility, the optimal portfolio is expressed into a unified formula consisting of four terms. The first term $\boldsymbol{\pi}^{(1)}$ is called the \textbf{Merton relative risk aversion term (RA)}, i.e., the classic constant percentage strategy in \cite{M1969}. 
The second term $\boldsymbol{\pi}^{(2)}$ is called the \textbf{risk seeking term (RS)}, which only appears on the linear part of the concave envelope ($R_k = 0$). As \cite{C2000} shows, the non-concave parts in the utility (or, the linear parts in the concave envelope) lead to a strong increase in the investment of the risky asset. The risk seeking term acts as the increase of the risky investment. Such an increase $\boldsymbol{\pi}^{(2)}$ is proportional to the length ($a_{k+1}- a_k$) of the linear parts in the concave envelope. The optimal percentage in the risky asset exceeds 100\% when time $t$ approaches the terminal time $T$ and the wealth amount lies in the linear parts. In Figure \ref{fig}, it happens when $X_t^* \in (a_1, a_2) \cup (a_2, a_3)$ and $t \rightarrow T$.  

The third term $\boldsymbol{\pi}^{(3)}$ is called the \textbf{loss aversion term (LA)}. There are some benchmarks in the utility, e.g., a wealth level $A \in \R$ for the HARA utility $U(x) = \frac{1}{1-R} (x-A)^{1-R}$ and a reference point for the S-shaped utility $U(x) = \frac{1}{1-R} (x-A)^{1-R} \id_{\{x \geq A\}} + (-\lambda) \frac{1}{1-R} (A-x)^{1-R} \id_{\{x < A\}}$. 
These benchmarks are regarded as deposit wealth and induce a decrease of the risky investment to avoid loss. $\boldsymbol{\pi}^{(3)}$ is actually a weighted sum of all discounted benchmarks (in Figure \ref{fig}, they are $A_0, A_3, A_4$). The weight is given by $q_k$.

The fourth term $\boldsymbol{\pi}^{(4)}$ is called the \textbf{first-order risk aversion term (First-order RA)}. \cite{P1964} and \cite{SS1997} discuss the different orders of risk aversion. \cite{SS1997} show that a non-differentiable point leads to a ``first-order" risk premium. Here we find that non-differentiability also causes a decrease $\boldsymbol{\pi}^{(4)}$ in the risky position. $\boldsymbol{\pi}^{(4)}$ is actually a weighted sum of all discounted non-differentiable points (in Figure \ref{fig}, they are $a_0, a_1, a_2, a_4$). The weight is given by $p_k$ (if $a_k$ is differentiable, $p_k = 0$). As a comparison, the Merton term $\boldsymbol{\pi}^{(1)}$ is actually caused by ``second-order" risk aversion; see \cite{P1964} for discussion on risk premium. Later in Section \ref{sec:numerical} we will find the decrease effect of $\boldsymbol{\pi}^{(4)}$ is much greater than that of $\boldsymbol{\pi}^{(3)}$.

Finally, we show a formula of the optimal portfolio for general case ($R_k \in [0, \infty]$) in Theorem \ref{thm_opt_strategy}.
\begin{theorem}[General case]
	\label{thm_opt_strategy}
	Suppose that the concave envelope $U^{**}$ has the form in Definition \ref{def_PHARA} and $R_k \in [0, \infty]$ for each $k \in \{0,1,\cdots, n\}$. For Problem \eqref{prob-main}, the optimal portfolio vector at time $t \in [0,T)$ is given by
	\begin{equation}\label{eq_general_strategy}
		\begin{aligned}
			\boldsymbol{\pi}^*_t &= (\boldsymbol{\sigma}^\top)^{-1}\boldsymbol{\theta} \sum_{k=0}^n \Bigg\{ R_k^{-1} X^R_{t,k} \id_{\{ R_k \neq 0,\infty \}} + e^{-r(T-t)}\frac{a_{k+1}-a_k}{||\boldsymbol{\theta}||_2\sqrt{T-t}}\Phi'\(d_1\(\frac{\gamma_{k}^{+}}{y^*\xi_t}\)\) \id_{\{ R_k = 0 \}} \\
			&\quad \quad \quad \quad \quad+ e^{-r(T-t)}\frac{1}{\alpha_k} \[\Phi\(d_1\(\frac{\gamma_{k+1}^{-}}{y^*\xi_t}\)\) - \Phi\(d_1\(\frac{\gamma_k^{+}}{y^*\xi_t}\)\)\] \id_{\{ R_k = \infty, A_k = -\infty, \alpha_k > 0 \}} \Bigg\}.
		\end{aligned}
	\end{equation}
\end{theorem}
This theorem shows a formula \eqref{thm_opt_strategy} of the optimal portfolio for the general case ($R_k \in [0, \infty]$), while $\boldsymbol{\pi}^*_t$ relies on $X^R_{t,k}$ defined in \eqref{eq:wealth_five_term}. Here $X^R_{t,k}$ is a part of $X_t^*$. Hence, for the more general case, this portfolio \eqref{eq_general_strategy} is not a feedback form of the wealth $X_t^*$.

\section{Examples revisiting}\label{sec:revisit}
The unified formula \eqref{eq:main} includes various closed-form portfolios in the literature as examples. We illustrate how to use the formula to directly write down the optimal portfolio by Examples \ref{ex:M}-\ref{ex:HK}. 

\begin{examplec}[ex:M]	
	Now we apply our formula \eqref{eq:main} to solve Problem \eqref{prob-main} in this example.
	\begin{itemize}
		\item 
		According to the expression of the CRRA utility $U$ in  \eqref{eq:U:CRRA}, we directly have the table of parameters: 
		\begin{table}[htbp]\centering
			\begin{tabular}{c|c|c|c|c|c|c|c}
				& $a_k$ & $\gamma_k^-$ & $\gamma_k^+$ & $R_k$ & $A_k$ & $q_k$ & $p_k$        \\
				\hline 
				$k=0$ & $0$   &  $\infty$    & $\infty$     & $R$ & $0$   &  $1$  & $0$   \\
				\hline
				$k=1$ & $\infty$ & $0$       & NA           & NA    & NA    & NA    & NA \\
				\hline
			\end{tabular}
			\captionof{table}{Parameter specification for the PHARA utility function \eqref{eq_def_U_tilde} and Theorem \ref{thm_main} ($U$ is PHARA: $n+1=1$) according to the notation of the CRRA utility in \cite{M1969}.}.\label{tab_CRRA}
		\end{table}
		
		Thus, according to \eqref{eq:main} in Theorem \ref{thm_main}, the optimal portfolio of Problem \eqref{prob-main} is given by
		\begin{equation}\label{eq:CRRA_strategy}
			\begin{aligned}
				\boldsymbol{\pi}_t^* = & (\boldsymbol{\sigma}^\top)^{-1} \boldsymbol{\theta} \frac{X_t^*}{R}, \; t \in [0, T],
			\end{aligned}
		\end{equation}
		which coincides with Equation (43) in \cite{M1969}. It is the well-known Merton portfolio with the percentage $(\boldsymbol{\sigma}^\top)^{-1} \boldsymbol{\theta} \frac{1}{R}$. We will discuss other portfolios with the Merton portfolio in Section \ref{sec:econ}.
		
		\item 
		According to the expression of the CARA utility $U$ in  \eqref{eq:U:CARA}, we directly have the table of parameters:
		\begin{table}[htbp]\centering
			\begin{tabular}{c|c|c|c|c|c|c|c|c}
				& $a_k$ & $\gamma_k^-$ & $\gamma_k^+$ & $\alpha_k$ & $R_k$ & $A_k$ & $q_k$ & $p_k$\\
				\hline 
				$k=0$ & $-\infty$ & $\infty$ & $\infty$ & $\alpha$ & $\infty$ & $-\infty$ & $1$ & $0$ \\
				\hline
				$k=1$ & $\infty$ & $0$ & NA & NA & NA & NA & NA & NA \\
				\hline
			\end{tabular}
			\captionof{table}{Parameter specification for the PHARA utility function \eqref{eq_def_U_tilde} and Theorem \ref{thm_main} ($U$ is PHARA: $n+1=1$) according to the notation  of the CARA utility in \cite{M1969}.}.\label{tab_CARA}
		\end{table}
		
		Thus, according to \eqref{eq_general_strategy} in Theorem \ref{thm_opt_strategy}, the optimal portfolio of Problem \eqref{prob-main} is given by
		\begin{equation}
			\begin{aligned}
				\boldsymbol{\pi}_t^* = & (\boldsymbol{\sigma}^{\top})^{-1} \boldsymbol{\theta}  \frac{e^{-r(T-t)}}{\alpha}, \; t \in [0, T).
			\end{aligned}
		\end{equation}
		which coincides with Equation (65) in \cite{M1969}.
	\end{itemize}
\end{examplec}

\begin{examplec}[ex:C]
	Now we apply our formula \eqref{eq:main} to solve Problem \eqref{prob-main} in this example. According to the expression of $U^{**}$ in  \eqref{eq:C_concavification}, we directly have the table of parameters:
	\begin{table}[htbp]\centering
		\begin{tabular}{c|c|c|c|c|c|c|c}
			& $a_k$ & $\gamma_k^-$ & $\gamma_k^+$ & $R_k$ & $A_k$ & $q_k$ & $p_k$        \\
			\hline 
			$k=0$ & $0$ &  $\infty$ & $U'(\hat{x})$ &  $0$   & NA  & $0$   & $1-\Phi\left(d_1\(\frac{U'(\hat{x})}{y^* \xi_t}\)\right)$  \\
			\hline
			$k=1$ & $\hat{x}$ & $U'(\hat{x})$ & $U'(\hat{x})$ & $1-\gamma$ & $B_T - \frac{K-\omega}{\alpha}$  & $\Phi\left(d_1\(\frac{U'(\hat{x})}{y^* \xi_t}\)\right)$   &  $0$ \\
			\hline
			$k=2$ & $\infty$ & $0$ & NA & NA & NA & NA & NA \\
			\hline
		\end{tabular}
		\captionof{table}{Parameter specification for the PHARA utility function \eqref{eq_def_U_tilde} and Theorem \ref{thm_main} ($U^{**}$ is PHARA: $n+1=2$, non-differentiable points: $a_0$, tangent point: $a_1$) according to the notation in \cite{C2000}. }.\label{tab:C2000}
	\end{table}
	
	\noindent
	Thus, according to \eqref{eq:main} in Theorem \ref{thm_main}, we directly have the optimal portfolio for both Problems \eqref{prob-main} and \eqref{prob:concavification}:
	\begin{equation}\label{eq:C_strategy}
		\begin{aligned}
			\boldsymbol{\pi}_t^* = & \underbrace{\frac{(\boldsymbol{\sigma}^\top)^{-1} \boldsymbol{\theta}}{1-\gamma} X_t^*}_{\text{Merton term}} +  \underbrace{\frac{e^{-r(T-t)}}{ \sqrt{T-t}} \frac{(\boldsymbol{\sigma}^\top)^{-1} \boldsymbol{\theta}}{||\boldsymbol{\theta}||_2} (a_1-a_0) \Phi'\(d_1\(\frac{\gamma_0^+}{y^* \xi_t}\)\)}_{\text{risk seeking}} - \underbrace{\frac{e^{-r(T-t)}(\boldsymbol{\sigma}^\top)^{-1} \boldsymbol{\theta}}{1-\gamma}  A_1 q_1}_{\text{loss aversion}}- \underbrace{\frac{e^{-r(T-t)}(\boldsymbol{\sigma}^\top)^{-1} \boldsymbol{\theta} }{1-\gamma} a_0 p_0,}_{\text{first-order risk aversion}}
		\end{aligned}
	\end{equation}
	where we use the market parameters of our current paper (Section \ref{sec:model}) as notation. We can see $\boldsymbol{\pi}^*_t$ in \eqref{eq:C_strategy} is exactly the optimal portfolio given by Equation (25) in \cite{C2000}.
\end{examplec}

\begin{examplec}[ex:LSW]
	Now we apply our formula \eqref{eq:main} to solve Problem \eqref{prob-main}. 
	%
	According to the expression of $U^{**}$ in  \eqref{eq:LSW_concavification}, we directly have the table of parameters in Table \ref{tab2}.  
	\begin{table}[htbp]\centering
		\scriptsize
		\begin{tabular}{c|c|c|c|c|c|c|c}
			& $a_k$ & $\gamma_k^-$ & $\gamma_k^+$ & $R_k$ & $A_k$ & $q_k$ & $p_k$        \\
			\hline 
			$k=0$ & $0$ &  $\infty$ & $K$ &  $0$   & NA  & $0$   & $\Phi\left(d_1\left(\frac{K}{y^* \xi_t}\right)\right)$  \\
			\hline
			$k=1$ & $\frac{\LT}{1-\gamma}$ & $K$ & $K$ & $1-\gamma$ & $L$   & $\Phi\left(d_1\left(\frac{\tilde{m}}{y^* \xi_t}\right)\right) - \Phi\left(d_1\left(\frac{K}{y^* \xi_t}\right)\right)$   &  $0$ \\
			\hline
			$k=2$ & $\frac{\LT}{\alpha}$ & $\tilde{m}$ & $(1-\delta \alpha)\tilde{m}$ & $1-\gamma$ & $\frac{1-\delta}{1-\delta \alpha} L$   & $1 - \Phi\left(d_1\left(\frac{(1-\delta \alpha)\tilde{m}}{y^* \xi_t}\right)\right)$  & $\Phi\left(d_1\left(\frac{(1-\delta \alpha)\tilde{m}}{y^* \xi_t}\right)\right) -\Phi\left(d_1\left(\frac{\tilde{m}}{y^* \xi_t}\right)\right)$\\
			\hline
			$k=3$ & $\infty$ & $0$ & NA & NA & NA   & NA  & NA \\
			\hline
		\end{tabular}
		\captionof{table}{Parameter specification for the PHARA utility function \eqref{eq_def_U_tilde} and Theorem \ref{thm_main} ($U^{**}$ is PHARA: $n=2$, non-differentiable points: $a_0$, $a_2$, tangent point: $a_1$). According to the notation in \cite{LSW2017}, we define $K:= \gamma \left( \LT/(1-\gamma) - \LT \right)^{\gamma - 1}$ and $\tilde{m} := \gamma \left( \LT/\alpha - \LT \right)^{\gamma-1}$. }.\label{tab2}
	\end{table}

	Thus, according to Theorem \ref{thm_main}, we directly have the optimal portfolio for both Problems \eqref{prob-main} and \eqref{prob:concavification}:
	\begin{equation}\label{eq:LSW_strategy}
		\boldsymbol{\pi}_t^* = \underbrace{\frac{\theta}{\sigma (1-\gamma)} X_t^*}_{\text{Merton term}} + \underbrace{\frac{e^{-r(T-t)}}{\sigma \sqrt{T-t}} (a_1 - a_0) \Phi'\left( d_1\left( \frac{\gamma_0^+}{y^* \xi_t} \right) \right)}_{\text{risk seeking}} - \underbrace{\frac{\theta e^{-r (T-t)}}{\sigma (1-\gamma)} \left( A_1 q_1  + A_2 q_2 \right)}_{\text{loss aversion}} - \underbrace{\frac{\theta e^{-r (T-t)}}{\sigma (1-\gamma)} (a_0 p_0 + a_2 p_2),}_{\text{first-order risk aversion}}
	\end{equation}
	where we use the market parameters of our current paper (Section \ref{sec:model}) as notation. 
	
	It is checked by tedious computation that $\boldsymbol{\pi}^*_t$ in \eqref{eq:LSW_strategy} is exactly the same as the optimal portfolio given by (34) in \cite{LSW2017}, where they use a totally different expression. Compared to their portfolio (34), the portfolio \eqref{eq:LSW_strategy} is directly given by our formula without complicated derivation and has a four-term division with clear economic meanings. 
\end{examplec}

\begin{examplec}[ex:HK]	
	Now we apply our formula \eqref{eq:main} to solve Problem \eqref{prob-main}. 
	According to the expression of $U^{**}$ in  \eqref{eq:HK_concavification}, we directly have the table of parameters: 
	\begin{table}[htbp]\centering
		\scriptsize
		\begin{tabular}{c|c|c|c|c|c|c|c}
			& $a_k$ & $\gamma_k^-$ & $\gamma_k^+$ & $R_k$ & $A_k$ & $q_k$ & $p_k$        \\
			\hline 
			$k=0$ & $b x_0$ &  $\infty$ & $U'((1+c^*) x_0)$ &  $0$   & NA  & $0$   & $\Phi\left(d_1\left(\frac{U'((1+c^*) x_0)}{y^* \xi_t}\right)\right)$  \\
			\hline
			$k=1$ & $(1+c^*) x_0$ & $U'((1+c^*) x_0)$ & $U'((1+c^*) x_0)$ & $1-p$ & $x_0$   & $1 - \Phi\left(d_1\left(\frac{U'((1+c^*) x_0)}{y^* \xi_t}\right)\right)$   &  $0$ \\
			\hline
			$k=2$ & $\infty$ & $0$ & NA & NA & NA   & NA  & NA \\
			\hline
		\end{tabular}
		\captionof{table}{Parameter specification for the PHARA utility function \eqref{eq:HK_concavification} and Theorem \ref{thm_main} ($U^{**}$ is PHARA: $n + 1 = 2$, non-differentiable points: $a_0$, tangent point: $a_1$) according to the notation in \cite{HK2018}.}.\label{tab:HK2018}
	\end{table}
	
	\noindent
	Thus, according to \eqref{eq:main} in Theorem \ref{thm_main}, we directly have the optimal portfolio for both Problems \eqref{prob-main} and \eqref{prob:concavification}:
	\begin{equation}\label{eq:HK_strategy}
		\begin{aligned}
			\pi_t^* = & \underbrace{\frac{\theta}{\sigma (1-p)} X_t^*}_{\text{Merton term}}  +  \underbrace{\frac{e^{-r(T-t)}}{\sigma \sqrt{T-t}} (a_1 - a_0) \Phi'\left(d_1\left(\frac{\gamma_0^+}{y^* \xi_t}\right)\right)}_{\text{risk seeking}} - \underbrace{\frac{\theta e^{-r(T-t)}}{\sigma(1-p)} A_1 q_1}_{\text{loss aversion}}- \underbrace{\frac{\theta e^{-r(T-t)}}{\sigma(1-p)} a_0 p_0,}_{\text{first-order risk aversion}}
		\end{aligned}
	\end{equation}
	where we use the market parameters of our current paper (Section \ref{sec:model}) as notation. We can see $\boldsymbol{\pi}^*_t$ in \eqref{eq:HK_strategy} is exactly the optimal portfolio given by Equation (3.4) in \cite{HK2018}.
\end{examplec}

\begin{table}[htbp]\centering
	\small
	\begin{tabular}{c | c | c | c | c | c }
		Literature & Context & RA & RS & LA & First-order RA\\
		\hline
		\cite{M1969} &          CRRA and CARA         & $\checkmark$ &  &   &     \\
		\hline
		\cite{C2000} &          Option payoff with HARA utilities   & $\checkmark$ & $\checkmark$ &  $\checkmark$ &  $\checkmark$   \\
		\hline
		\cite{BKP2004} & Loss aversion case           & $\checkmark$ & $\checkmark$ & $\checkmark$ &    \\
		\hline
		\cite{BKP2004} & Kinked power utility case    & $\checkmark$ &  &  $\checkmark$ & $\checkmark$  \\
		\hline
		\cite{LSW2017} & Participating insurance contracts & $\checkmark$ & $\checkmark$ & $\checkmark$ & $\checkmark$  \\
		\hline
		\cite{HK2018} & First-Loss schemes in hedge funds    & $\checkmark$ & $\checkmark$ & $\checkmark$ & $\checkmark$   \\
		\hline
		\cite{HLLM2019} & Incentive schemes in pension funds    & $\checkmark$ & $\checkmark$ &  $\checkmark$ & $\checkmark$  \\
		\hline
		\cite{LL2020} & Principal's constraint   & $\checkmark$ & $\checkmark$ &  $\checkmark$ &   \\
		\hline
	\end{tabular}
	\captionof{table}{Unifying the literature}\label{tab1}
\end{table}

Table \ref{tab1} summarizes the portfolios above and in the literature in terms of four-term division. In the above and other literature of continuous-time portfolio selection, the type of optimal investment problems involve long computation and complex expression in the portfolio. In the following Section \ref{sec:econ}, we will see that our formula \eqref{eq:main} not only serves to ease computation, but also provides analytical tractability for financial analysis with a delicate four-term division.



%
%
%
%

\section{Economic meanings}\label{sec:econ}
We illustrate the economic meanings of our optimal portfolio formula for PHARA utilities in three aspects. First, we conduct an asymptotic analysis to the optimal portfolio $\boldsymbol{\pi}_t^*$ and the optimal wealth $X_t^*$ in terms of the pricing kernel $\xi_t$ for $t \in (0, T)$. Second, we numerically plot and illustrate the meaning of the PHARA utility. Finally, we compare the optimal portfolios under the PHARA utility (using Example \ref{ex:PHARA_complex}) with other utility families  analytically and numerically. We elaborate special features of the PHARA portfolio.

\subsection{Asymptotic analysis}\label{sec:asymptotic}
We conduct an asymptotic analysis for the general PHARA utility to illustrate financial insights in Theorem \ref{thm_option table}, following the ideas of \cite{C2000} and \cite{LL2023}. As $\xi_t$ is an indicator of the market state, by asymptotic analysis we know the risk-taking behavior of the portfolio. Theorem \ref{thm_option table} serves a general version of asymptotic analysis and adopts detailed limiting techniques. One can directly apply Theorem \ref{thm_option table} for asymptotic financial analysis given a specific PHARA utility. The proof is stated in Appendix \ref{app:proof}. 


\begin{theorem}	\label{thm_option table}
	Suppose the setting in Theorem \ref{thm_main} holds. Fix $t \in (0,T)$ and $y^* \in (0, \infty)$. 
	\begin{enumerate}[(i)]
		
		\item \label{thm_table_1} 
		As $\xi_t \rightarrow 0$, we have
		\begin{equation}		
			X_t^*  \rightarrow \infty, \;\;
			\boldsymbol{\pi}_t^{*}  \rightarrow \boldsymbol{\infty},\;\;
			\frac{\boldsymbol{\pi}_t^{*}}{X_t^*}  \rightarrow \frac{1}{R} (\boldsymbol{\sigma}^\top)^{-1} \boldsymbol{\theta}, 
		\end{equation}
		and more detailed results for different terms are given in Tables \ref{tab:thm_asymptotic1}-\ref{tab:thm_asymptotic2}.
		\item \label{thm_table_2}
		As $\xi_t \rightarrow \infty$, we have
		\begin{equation}
			\frac{X_t^*}{a_0 e^{-r(T-t)}} \rightarrow 1, \;\;
			\boldsymbol{\pi}_t^{*}  \rightarrow \mathbf{0},\;\;
			\frac{\boldsymbol{\pi}_t^{*}}{X_t^*}  \rightarrow \mathbf{0},
		\end{equation}
		and more detailed results for different terms are given in Tables \ref{tab:thm_asymptotic1}-\ref{tab:thm_asymptotic2}.
		
		\begin{table}[htbp]	
			\centering
			\begin{tabular}{c|c|c|c|c}
				\hline
				& $X_t^*$  & $X_t^D$  & $X_t^A$  & $X_t^R$   \\
				\hline
				$\xi_t \rightarrow 0$ & $\infty$ & 0 & $A_n e^{-r(T-t)}$ & $\infty$  \\
				\hline
				$\xi_t \rightarrow \infty$, $\gamma_0^+ = \infty$ & $a_0 e^{-r(T-t)}$  & 0 & $a_0 e^{-r(T-t)}$ & $0$  \\
				\hline
				$\xi_t \rightarrow \infty$, $\gamma_0^+ < \infty$ & $a_0 e^{-r(T-t)}$  & $a_0 e^{-r(T-t)}$ & 0 & $0$  \\
				\hline
			\end{tabular}
			\caption{Dynamic analysis for $X_t^*$.}
			\label{tab:thm_asymptotic1}
		\end{table}
		\begin{table}[htbp]	
			\centering
			\begin{tabular}{c|c|c|c|c|c}
				\hline
				&  $\boldsymbol{\pi}_t^{*}$ &$\frac{\boldsymbol{\pi}_t^{*}}{X_t^*}$  &  $\boldsymbol{\pi}_t^{(2)}$  & $\boldsymbol{\pi}_t^{(3)}$ & $\boldsymbol{\pi}_t^{(4)}$   \\
				\hline
				$\xi_t \rightarrow 0$ & $\boldsymbol{\infty}$  & $\frac{1}{R}(\boldsymbol{\sigma}^\top)^{-1} \boldsymbol{\theta}$ & $\mathbf{0}$ & $-\frac{1}{R} A_n e^{-r(T-t)} (\boldsymbol{\sigma}^\top)^{-1} \boldsymbol{\theta} $ & $\mathbf{0}$  \\
				\hline
				$\xi_t \rightarrow \infty$, $\gamma_0^+ = \infty$ & $\mathbf{0}$ & $\mathbf{0}$ & $\mathbf{0}$ & $-\frac{1}{R} a_0 e^{-r(T-t)} (\boldsymbol{\sigma}^\top)^{-1} \boldsymbol{\theta}$ & $\mathbf{0}$ \\
				\hline
				$\xi_t \rightarrow \infty$, $\gamma_0^+ < \infty$ & $\mathbf{0}$ & $\mathbf{0}$ & $\mathbf{0}$ & $\mathbf{0}$ & $-\frac{1}{R} a_0 e^{-r(T-t)} (\boldsymbol{\sigma}^\top)^{-1} \boldsymbol{\theta}$ \\
				\hline
			\end{tabular}
			\caption{Dynamic analysis for $\boldsymbol{\pi}^*_t$.}
			\label{tab:thm_asymptotic2}
		\end{table}
	\end{enumerate}
\end{theorem}
We first analyze the wealth process. When the market state is good ($\xi_t \rightarrow 0$), we see that the wealth is going to infinity, which means that the portfolio performs well and leads to an increasing wealth process. The increase is mainly due to $X_t^R$, caused by the strictly concave CRRA parts in the utility. The CRRA part on $(a_n, \infty)$ also provides an amount of wealth $X_{t,n}^A$, which approaches a positive term $A_n e^{-r(T-t)}$. In this scenario, the other wealth terms $X_{t,k}^A$, $k = 0, \cdots, n-1$ and $X_t^D$ are negligible and approach 0. When the market state is bad ($\xi_t \rightarrow \infty$), we see that the wealth is going to a lower bound ($a_0 e^{-r(T-t)}$), which means that the portfolio performs bad and lies at the least possible level (otherwise the utility value will be negative infinity). This lower bound level is also known as liquidation boundary; see \cite{HK2018}. The amount of wealth is contributed by $X_{t,0}^D$ (if $\gamma_0^+ < \infty$) or $X_{t,0}^A$ (if $\gamma_0^+ = \infty$). The former one $X_{t,0}^D$ appears because the left- and right- derivatives at $a_0$ are not equal, which makes the decision maker keep the least amount of wealth $a_0 e^{-r(T-t)}$ to avoid liquidation. When $\gamma_0^- = \gamma_0^+ = \infty$, the Loss Aversion term $X_{t,0}^A$ appears because the utility is CRRA on $(a_0, a_1)$. This will provide a same amount of deposit wealth for the benchmark $A_0 e^{-r(T-t)} = a_0 e^{-r(T-t)}$.

Now we discuss the optimal portfolio. When the market state is good ($\xi_t \rightarrow 0$), the optimal risky investment percentage approaches the Merton percentage $\frac{1}{R} (\boldsymbol{\sigma}^\top)^{-1} \boldsymbol{\theta}$. In a good state, the decision maker does not gamble and the RS term $\boldsymbol{\pi}_t^{(2)}$ approaches zero. By definition, the LA term $\boldsymbol{\pi}_t^{(3)}$ is actually a weighted sum of all discounted benchmarks. As $\xi_t \rightarrow 0$, the weight $q_n$ approaches 1 and the only benchmark left in a good state is the largest one $A_n$, while all other weights $q_k$, $k = 0, \cdots, n-1$, approach 0. As a result, $\boldsymbol{\pi}_t^{(3)}$ approaches $-\frac{1}{R} (\boldsymbol{\sigma}^\top)^{-1} \boldsymbol{\theta} A_n$. In the good state, the First-order RA term asymptotically becomes zero and all weights of the non-differentiable benchmarks vanish similarly. As a whole, the only benchmark which $\boldsymbol{\pi}_t^{(3)} + \boldsymbol{\pi}_t^{(4)}$ tracks in the good state is the largest (richest) one, $A_n e^{-r(T-t)}$.  
When the market state is bad ($\xi_t \rightarrow \infty$), the optimal risky investment percentage approaches $\mathbf{0}$. It means that to invest all on the risk-free asset is optimal. Because the pricing kernel process has a rather large value, it is less likely to make money by gambling on the risky assets. Hence, the RS term $\boldsymbol{\pi}_t^{(2)}$ approaches $\mathbf{0}$. The Merton term $\boldsymbol{\pi}_t^{(1)}$ contributes to a risky investment $\frac{1}{R}(\boldsymbol{\sigma}^\top)^{-1} \boldsymbol{\theta} a_0 e^{-r(T-t)}$, while either the LA term $\boldsymbol{\pi}_t^{(3)}$ or the first-order RA term $\boldsymbol{\pi}_t^{(4)}$ deducts this risky investment in cases that $\gamma_0^+ = \infty$ or $\gamma_0^+ < \infty$. As a whole, the only benchmark which $\boldsymbol{\pi}_t^{(3)} + \boldsymbol{\pi}_t^{(4)}$ tracks in the bad state is the smallest (poorest) one, $a_0 e^{-r(T-t)}$. As a result, the total risky investment $\boldsymbol{\pi}_t$ asymptotically approaches $\mathbf{0}$, making the optimal wealth process evolve along the lower bound $a_0 e^{-r(T-t)}$. 

\subsection{Numerical illustration}\label{sec:numerical}
We illustrate the economic meaning of the PHARA utility $U$ by investigating the optimal portfolio in Example \ref{ex:PHARA_complex} and Figure \ref{fig}. Note that here $m=1$, and there is no bold symbol. We denote the one-dimensional optimal portfolio by $\pi^*$. First, we can directly write the optimal portfolio from \eqref{eq:main}:
\begin{equation}\label{eq:example}
	\begin{aligned}
		{\pi}^*_t &= \underbrace{\frac{\theta}{\sigma R} X^*_t}_{\text{Merton term}} + \underbrace{\frac{e^{-r(T-t)}}{\sigma \sqrt{T-t}} 
			\left( (a_2-a_1) \Phi'\(d_1\(\frac{\gamma_1^{+}}{y^* \xi_t}\)\) + (a_3-a_2) \Phi'\(d_1\(\frac{\gamma_2^{+}}{y^* \xi_t}\)\) \right)}_{\text{risk seeking}}\\
		&\quad \quad \quad \underbrace{-\frac{\theta e^{-r(T-t)}}{\sigma R} \left(A_0 q_0 + A_3 q_3 + A_4 q_4 \right)}_{\text{loss aversion}}
		\underbrace{-\frac{\theta e^{-r(T-t)}}{\sigma R} \left(a_0 p_0 + a_1 p_1 + a_2 p_2 + a_4 p_4\right) }_{\text{first-order risk aversion}}\\
		&:= {\pi}_t^{(1)} + {\pi}_t^{(2)} + {\pi}_t^{(3)} + {\pi}_t^{(4)}.
	\end{aligned}
\end{equation}
In Figure \ref{fig2}, we numerically demonstrate the optimal portfolio \eqref{eq:example} for the PHARA utility $U$ of Figure \ref{fig}. The optimal portfolios under a HARA utility and a SAHARA utility (symmetric asymptotic HARA; proposed by \cite{CPV2011}) are also plotted. We will later discuss and compare these portfolios in Section \ref{sec:comparison}.
\begin{figure}[htbp]
	\centering
	\includegraphics[width=0.95\textwidth]{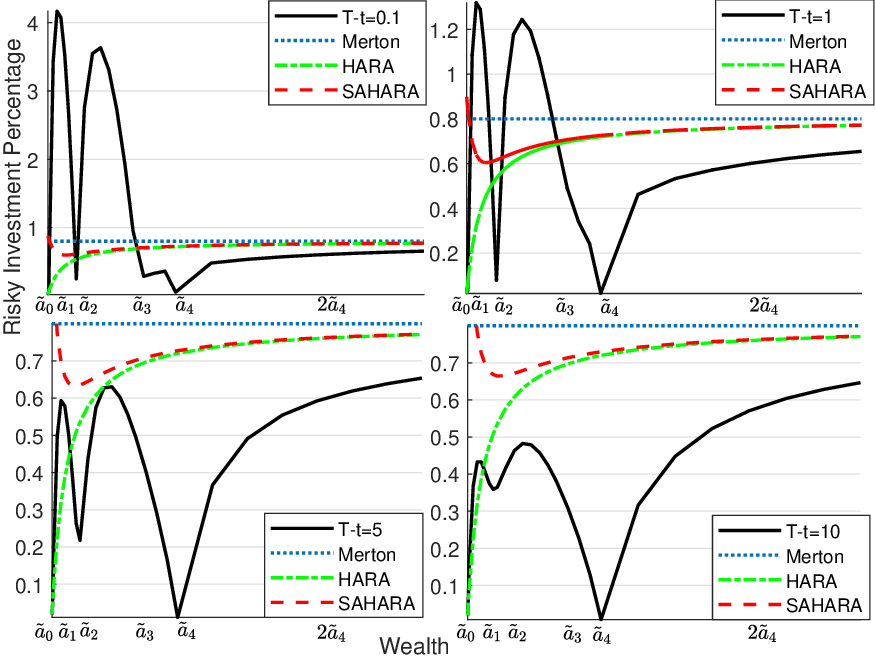}
	\caption{Comparison among portfolios ($y$-axis: ``Risky Investment Percentage" means the optimal percentage invested in the risky asset ${\pi}^*_t/X_t^*$, $x$-axis: ``Wealth" means $X_t^*$). The black solid line is the  portfolio \eqref{eq:example} for the PHARA utility of Figure \ref{fig}. The other portfolios are Merton \eqref{eq:M_strategy}, HARA \eqref{eq:HARA_strategy} and SAHARA \eqref{eq:SAHARA_strategy}; to be consistent with \eqref{eq:example}, the parameter setting of the HARA utility \eqref{eq:HARA} is given by $a_0 = 4$ and $R = 0.5$ and the SAHARA utility \eqref{eq:SAHARA} is given by $\psi = R = 0.5$ and $v = a_0 = 4$; \textbf{here $\beta = 20$}.   
		The $x$-axis limits are scaled by $e^{-r (T-t)}$, and $\tilde{a}_k := a_k e^{-r (T-t)}$ for $k = 0, 1, \cdots, 4$. The market setting includes $r = 0.05$, $\sigma = 0.3$, $\mu = 0.086$ (hence, $\theta = (\mu-r)/\sigma = 0.12$) and $R = 0.5$. The Merton percentage is $\frac{\theta}{\sigma R} = 0.8$.}
	\label{fig2}
\end{figure}

By numerical illustration and analytical analysis to \eqref{eq:example}, we obtain the following economic meanings, some of which is also shown in asymptotic analysis in Section \ref{sec:asymptotic}. Here the effect of non-differentiable points is our main novel finding.
\begin{enumerate}[(a)]
	\item Generally, the contour of the optimal portfolio is ``\textbf{multiple-peak-multiple-valley}" with the tail trend to the Merton term.
	
	\item \textbf{Peak on linearity}: The optimal portfolio gambles drastically on the linear parts (here, $(a_1, a_2)$, $(a_2, a_3)$) of the concave envelope $U^{**}$. Especially, the risky investment percentage exceeds 100\% if it approaches to the terminal time. From the formula \eqref{eq:example}, we know that the percentage actually tends to infinity with a rate $(T-t)^{-1/2}$ when $t \rightarrow T$. The gambling effect is because of the risk seeking term ${\pi}^{(2)}$.  
	
	\item \textbf{Valley at non-differentiability}: The optimal portfolio becomes extremely conservative at the non-differentiable points (here, $a_0, ~a_1, ~a_2, ~a_4$) of the concave envelope $U^{**}$. 
	We see that the percentage actually tends to zero at non-differentiable points. The conservative effect is because of the first-order risk aversion term ${\pi}^{(4)}$.
	
	\item \textbf{Merton trend}: The optimal risky investment percentage tends to the Merton percentage (here, $\frac{\theta}{\sigma R} = 0.8$) if the wealth is large enough. In this case, for $i= 2,3, 4$, we have ${\pi}^{(i)}_t/X_t^* \rightarrow 0 \;(X_t^* \rightarrow \infty)$. As a result, the trend is due to ${\pi}^{(1)}$.
	
	\item The decreasing effect of the loss aversion term ${\pi}^{(3)}$ at the benchmark points (here, $A_3, ~A_4$) is slight. At these points, surprisingly, the decision maker is still taking great risk. 
\end{enumerate}

\subsection{Portfolio comparison}\label{sec:comparison}

In Figure \ref{fig2}, we also plot the strategies under the HARA utility and the SAHARA utility proposed by \cite{CPV2011}. Analytically, the SAHARA (symmetric asymptotic HARA) utility family is also a generalization of the HARA family as the PHARA utility. The SAHARA utility characterizes a special case of non-monotone absolute risk aversion (ARA) functions, where the utility $U$ satisfies 
\begin{equation}\label{eq:SAHARA}
	\text{ARA}(x) := -\frac{U''(x)}{U'(x)} = \frac{1}{\sqrt{\(\frac{\beta}{\psi}\)^2 + \(\frac{x-v}{\psi}\)^2}},\;\; x \in \R,
\end{equation}
where $\psi > 0$, $\beta > 0$ and $v \in \R$ are parameters. According to Proposition 2.2 of \cite{CPV2011}, if $v = 0$, the expression of a SAHARA utility has a domain $\R$ and is given by 
\begin{equation}
	U(x) = c_1 + c_2 \tilde{U}(x),
\end{equation}
where  
\begin{equation}
	\tilde{U}(x) := \left\{
	\begin{aligned}
		& -\frac{1}{\psi^2 - 1} \( x + \sqrt{\beta^2 + x^2} \)^{-\psi} \( x + \psi \sqrt{\beta^2 + x^2} \), && \text{ if } \psi \neq 1;\\
		& \frac{1}{2} \log \( x + \sqrt{\beta^2 + x^2} \) + \frac{1}{2} \beta^{-2} x \( \sqrt{\beta^2 + x^2} - x \), && \text{ if } \psi = 1,
	\end{aligned}
	\right.\;\;\;\; x \in \R, 
\end{equation}
where $c_1, c_2 \in \R$ are constants; we simply set $c_1 = 0$ and $c_2 = 1$. If $\beta = 0$, the SAHARA utility reduces to the HARA utility. In our paper, on the contrary, the PHARA utility focuses on the different risk aversion on different parts of the domain. On each part, the ARA function is hyperbolic, i.e., for $k \in \{0, 1, \cdots, n\}$, 
for any $x \in(a_k, a_{k+1})$, 
\begin{equation}\label{eq:ARA}
	\begin{aligned}
		\text{ARA}(x) := -\frac{U''(x)}{U'(x)} =  \frac{R_k}{x - A_k}
		= \left\{
		\begin{aligned}
			& 0, &&\text{ if } R_k = 0;\\
			& \alpha_k, &&\text{ if } R_k = \infty, A_k = -\infty,  \alpha_k \in (0, \infty);\\
			& \frac{1}{\frac{1}{R_k}x - \frac{A_k}{R_k}}, &&\text{ if } R_k \in (0, \infty),
		\end{aligned}
		\right.
	\end{aligned}
\end{equation}
where $A_k, R_k, \alpha_k$ are given in Definition \ref{def_PHARA}. It is clear that the ARA function \eqref{eq:ARA} is locally monotone (decreasing) but not globally monotone. As a result, the two families of PHARA and SAHARA utility functions lead to very different optimal portfolios. 

In the one-dimensional Black-Scholes model ($m=1$), recalling the CRRA portfolio \eqref{eq:CRRA_strategy} in Example \ref{ex:M}, \cite{M1969}'s optimal portfolio for 
\begin{equation}
	U(x) = \frac{1}{1-R}x^{1-R}, \;\; x \in (0, \infty), 
\end{equation}
is given by 
\begin{equation}\label{eq:M_strategy}
	\pi_t^{\text{Merton}} = \frac{\theta}{\sigma R} X_t^*, \;\; t \in [0, T],    
\end{equation}
which is proportional to the corresponding optimal wealth $X_t^*$. For the HARA utility, using the convention in Definition \ref{def_PHARA}, we write 
\begin{equation}\label{eq:HARA}
	U(x) = \frac{1}{1-R} (x-a_0)^{1-R}, \;\; x \in (a_0, \infty),
\end{equation}
where $a_0$ is a constant in $\R$. According to Theorem \ref{thm_main}, the portfolio for \eqref{eq:HARA} is directly given by 
\begin{equation}\label{eq:HARA_strategy}
	\pi_t^{\text{HARA}} = \frac{\theta}{\sigma R} \( X_t^* - a_0 e^{-r(T-t)}\),~~ t \in [0, T].
\end{equation}
The optimal portfolio \eqref{eq:HARA_strategy} is a ratio of the optimal wealth (Merton term) plus a LA correction. For the PHARA utility, the optimal portfolio \eqref{eq:main} is a ratio of the optimal wealth (Merton term) plus three corrections (RS, LA, First-order RA) based on different parts of the utility function. Finally, the SAHARA utility leads to the optimal portfolio 
\begin{equation}\label{eq:SAHARA_strategy}
	\pi_t^{\text{SAHARA}} = \frac{\theta}{\sigma \psi} \sqrt{\(X_t^* - v e^{-r(T-t)}\)^2 + b_t^2}, ~~ t \in [0, T],
\end{equation}
where $b_t := \beta e^{-(r-\frac{\theta^2}{2\psi^2})(T-t)}$. This portfolio does not contain a proportional Merton term anymore. It is due to the non-monotone feature of the ARA function \eqref{eq:SAHARA}. Therefore, there is an essential discrepancy between PHARA and SAHARA utility families.

\subsection{Numerical comparison}\label{sec:nume_comparison}
In Figure \ref{fig2}, the portfolios \eqref{eq:HARA_strategy} and \eqref{eq:SAHARA_strategy} under the HARA utility and the SAHARA utility are plotted and compared with the proposed portfolio \eqref{eq:example}. 

Above all, all these portfolios share the Merton trend if the wealth is large enough. The HARA portfolio is a simplified version of the PHARA portfolio. When the market state is good ($\xi_t \rightarrow 0$), the asymptotic behavior of the HARA portfolio is to invest in a Merton term. When the market state is bad ($\xi_t \rightarrow \infty$), that of the HARA portfolio is to invest all in the risk-free asset. These two features are the same in the PHARA portfolio. In between, the HARA portfolio is to simply increase the risky investment with respect to the wealth. The PHARA portfolio has a  sophisticated behavior: peak on linearity (here, $(a_1, a_2)$, $(a_2, a_3)$) and valley at non-differentiability (here, $a_0$, $a_1$, $a_2$, $a_4$), as discussed in Section \ref{sec:numerical}.  

\begin{figure}[htbp]
	\centering
	\includegraphics[width=0.95\textwidth]{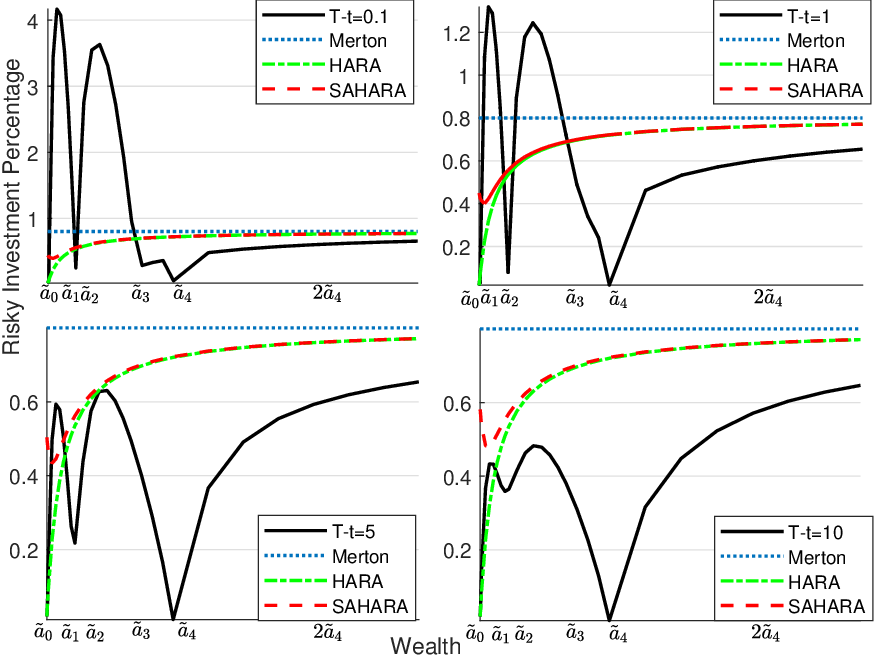}
	\caption{Comparison among portfolios ($y$-axis: ``Risky Investment Percentage" means the optimal percentage invested in the risky asset ${\pi}^*_t/X_t^*$, $x$-axis: ``Wealth" means $X_t^*$). All the setting is the same as Figure \ref{fig2} expect that \textbf{here $\beta = 5$ in the SAHARA utility \eqref{eq:SAHARA}}. }
	\label{fig2-beta5}
\end{figure}

For the SAHARA utility, the key parameter is $\beta > 0$. A feature of the SAHARA portfolio is to invest more in the risky assets when the market state is bad and the wealth is low, i.e., to gamble back at a recession state. When the market state is good, the SAHARA portfolio becomes conservative and decreases the risky investment. We additionally plot Figure \ref{fig2-beta5} with the only difference $\beta = 5$. As $\beta$ is smaller is Figure \ref{fig2-beta5}, the SAHARA portfolio becomes more similar to the HARA portfolio. This gambling effect at a bad state is mitigated in Figure \ref{fig2} when $\beta$ is small. 

In conclusion, the PHARA portfolio is more subtle than the HARA portfolio and the latter is more subtle than the Merton portfolio, while the SAHARA portfolio holds an investment logic opposite to the HARA portfolio to some degree. 

\section{Real-data study}
\label{sec:Real-data study}
We provide simulation results for the PHARA portfolio based on parameters collected empirically from the real-world financial data. We set $m=1$ in the following setting. 


We analyze the PHARA strategy \eqref{eq:example} over a time span of 10 years from 10/04/2013 to 10/04/2023. The first 8 years serve as the in-sample period. The latest two years serve as the out-of-sample period. Our investment time horizon, $T$, is set to be equivalent to the out-of-sample period of two years. The initial wealth, $x_0$, is arbitrarily set to $\$10$, and the risk free interest rate, $r$, is set to $0.05$ to capture the current two-year treasury rate. 
We select the SPDR S\&P 500 ETF Trust (SPY) to be our single asset. The S\&P 500 is a widely-accepted benchmark index and so we choose to trade an ETF that tracks this index. 
The estimated standard deviation of daily returns for the asset is denoted by $\sigma$. The mean of daily returns, the expected return rate, is denoted by $\mu$. We use the in-sample data to estimate $\mu = 0.27$ and $\sigma = 0.24$. Both $\mu$  and $\sigma$ are annualized to the two-year time horizon. This means that the average daily return is multiplied by 504 and the standard deviation of daily returns is multiplied by $\sqrt{504}$; here 504 is the number of trading days in two years.

\begin{figure}[htp]
	\centering
	\includegraphics[width=10cm]{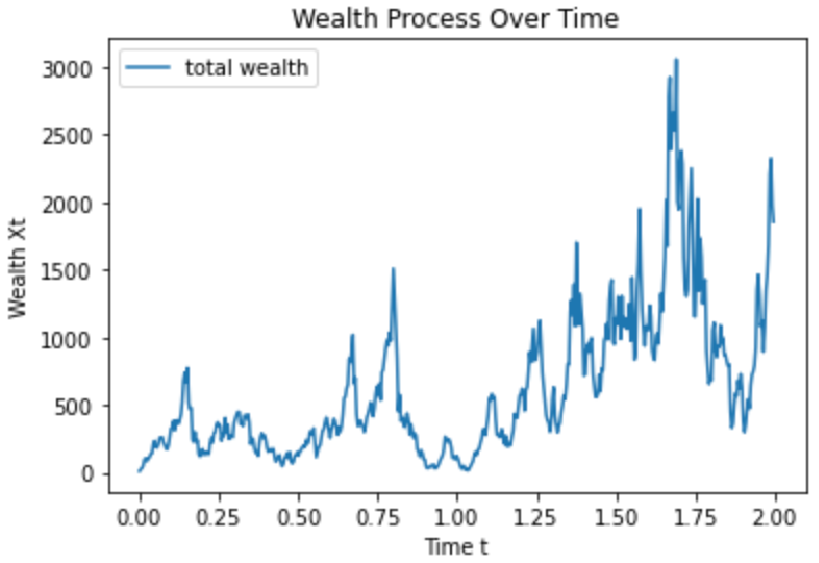}
	\caption{A simulation of the evolution of total wealth $X_t$ over the two year investing period}
	\label{fig:fig-4}
\end{figure}

Figure \ref{fig:fig-4} provides a snapshot example of one simulation of the strategy \eqref{eq:example}'s performance for the one dimensional case. We can see a high degree of volatility but an overall positive outcome in this particular simulation. The first measure of performance we utilize is the Sharpe Ratio; see, e.g., \cite{S1994}. Let $\text{po}$ denote the set of daily returns of the optimal portfolio and $R_{\text{po}}$ denote the average of this set. Additionally, let $R_{\text{fr}} = \frac{r}{504}$ denote the daily risk-free interest rate. Finally, let $\text{ex}$ denote the set of excess daily returns of the optimal portfolio. The Sharpe Ratio $S_a$ is given by
\begin{equation}
	S_a = \frac{R_{\text{po}} - R_{\text{fr}}}{\sigma_{\text{ex}}},
\end{equation}
where $\sigma_{\text{ex}}$ is the standard deviation of the set of excess daily returns of the optimal strategy. Figure \ref{fig:fig-5} provides a histogram of simple returns and Sharpe Ratios over 1000 simulations. Note that outliers (returns greater or less than three standard deviations from the mean) are removed from the simple returns histogram. Additionally, note that the final Sharpe Ratios are annualized to the two-year out-of-sample period.

\begin{figure}
	\centering
	\begin{minipage}{0.45\textwidth}
		\includegraphics[width=\linewidth]{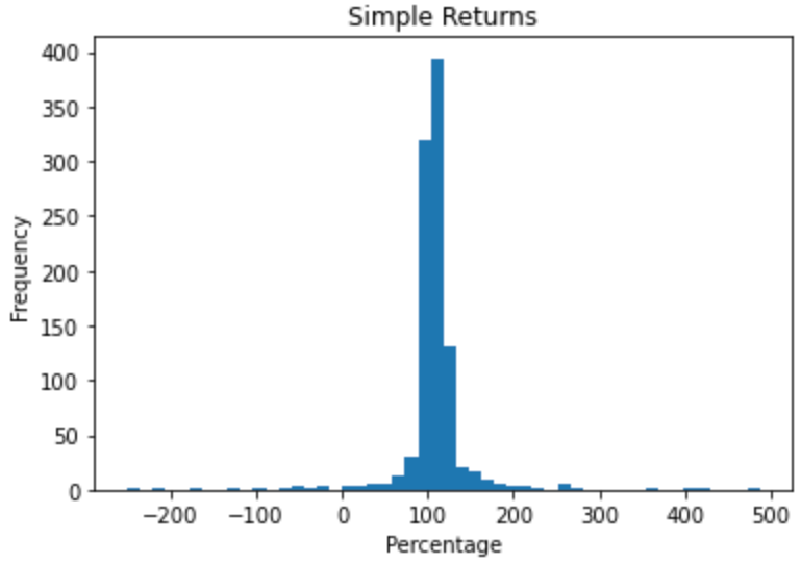}
	\end{minipage}\hfill
	\begin{minipage}{0.45\textwidth}
		\includegraphics[width=\linewidth]{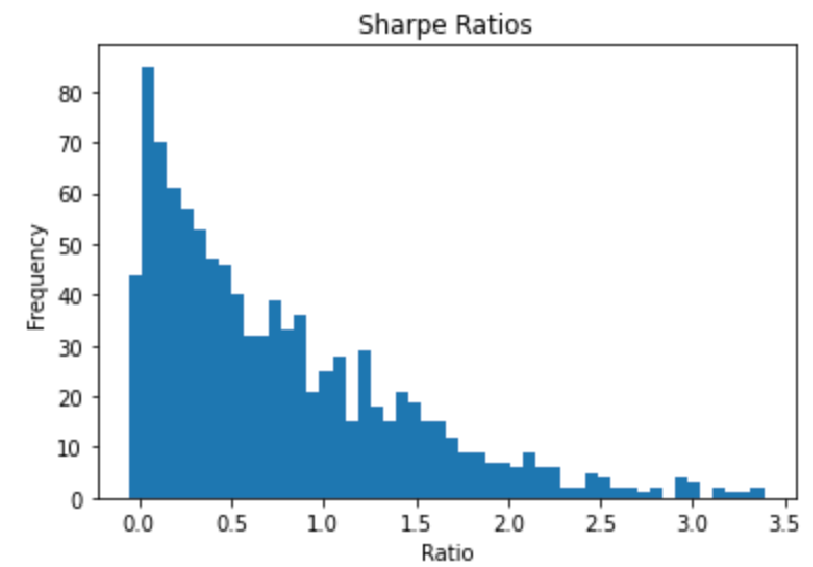} 
	\end{minipage}
	\caption{Histogram of simple returns and Sharpe Ratios for each of the 1000 simulations}
	\label{fig:fig-5}
\end{figure}

The second measure of performance we utilize is alpha or the capital asset pricing model \eqref{eq:alpha}; see, e.g.,  \cite{P2004}. We use the S\&P 500 as a benchmark index representative of the market and calculate daily returns for the out-of-sample period; this set of returns is denoted by $\text{ma}$. We also utilise $R_{\text{po}}$ and $R_{\text{fr}}$ in the calculation of $\alpha$:
\begin{equation}\label{eq:alpha}
	\alpha = R_{\text{po}} - R_{\text{fr}} - \beta(R_{\text{ma}} - R_{\text{fr}}),
\end{equation}
where $R_{\text{ma}}$ is the mean of the set of returns $\text{ma}$, and the systematic risk $\beta := \cov(\text{po}, \text{ma}) / \sigma_{\text{ma}}^2$ is given by the covariance of daily returns between the portfolio strategy and the market benchmark divided by the variance of the benchmark daily returns. The final alpha values are also annualized to the two-year out-of-sample period.

\begin{figure}[htp]
	\centering
	\includegraphics[width=10cm]{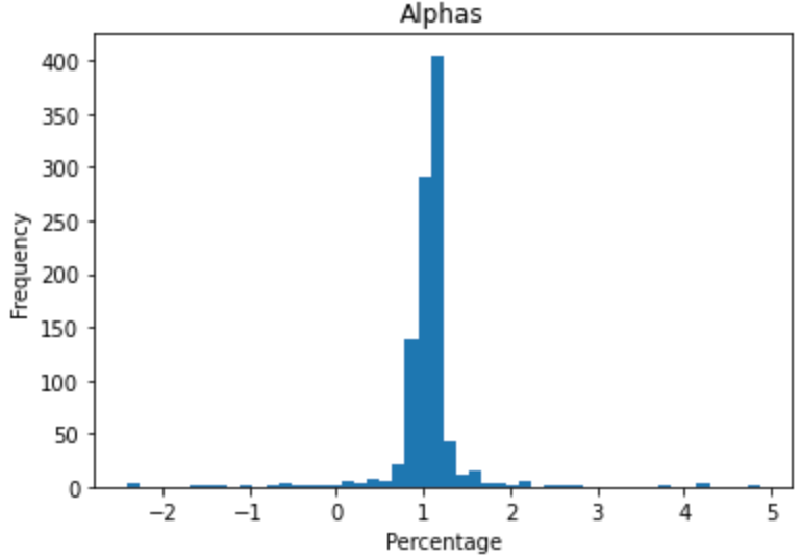}
	\caption{Alpha value distribution for 1000 simulations}
	\label{fig:fig-6}
\end{figure}

Figure \ref{fig:fig-6} shows the distribution of the calculated alpha values for the simulations. Note that outliers are also removed from the histogram of alphas. We can see the majority of the distribution is greater than 0\%. This shows us that we generate positive alpha with our strategy. Table \ref{tab:summary_stats} displays the mean and standard deviation for simple returns and our two measures of performance. From the summary statistics of simple returns, we see that the current PHARA portfolio is highly volatile yet clearly successful. From Figure \ref{fig:fig-5}, we can see modality around a simple return of $100 \%$. Additionally, Figure \ref{fig:fig-5} shows us that Sharpe Ratio returns are almost all positive but mostly situated below 1.0. The Sharpe Ratio histogram is also clearly right skewed. Traditionally, a Sharpe Ratio greater than 1.0 is regarded as a good performance. Our lower mean Sharpe Ratio demonstrates the fact that high volatility is a drawback of the PHARA portfolio. On the other hand, the high mean simple returns and positive mean alpha show that the PHARA portfolio provides high returns and better performance than market benchmarks. Specifically, the positive mean alpha shows that our strategy can outperform the market. Given this evidence, we conclude that our strategy may provide incredibly high rewards for a high volatility or risk.

\begin{table}[htbp]	
	\centering
	\begin{tabular}{c|c|c|c}
		\hline
		& Simple Returns  & $S_a$  & $\alpha$   \\
		\hline
		Mean & 115.95 & 0.75 & 1.1 \\
		\hline
		Standard Deviation & 157.56 & 0.68 & 1.64 \\
		\hline
	\end{tabular}
	\caption{Mean and standard deviation of measures of performance}
	\label{tab:summary_stats}
\end{table}

\section{Conclusion}\label{sec:conclusion}
We derive a unified formula for the optimal portfolio of the proposed piecewise hyperbolic absolute risk aversion (PHARA) utility family. In the field of asset allocation (e.g., hedge fund management), this formula can be directly applied to the portfolio strategy of a given utility function. The terms of the portfolio are illustrated with clear economic meaning on the risky behavior. We propose a general asymptotic analysis and find that non-concavity in the utility function causes great risk taking, while non-differentiability greatly reduces the risky behavior. The PHARA portfolio is featured with a contour of  ``multiple-peak-multiple-valley" with a tail trend to the Merton percentage, compared to the portfolios of HARA and SAHARA utilities. As a financial application, we empirically observe the performance of the PHARA portfolio over a two year period. We see that the PHARA portfolio is able to generate high returns and positive alpha, yet one of the drawbacks of the PHARA portfolio is high volatility. We can use the PHARA utilities as building blocks to approximate the optimal portfolio in the cases that only information on concavity and differentiablity of some intervals is available, which is left for a future direction.  

\vskip -0.2cm
\quad 

\noindent
{\bf Acknowledgements.}
The authors are grateful to the editors and two anonymous referees for valuable comments that have greatly improved this paper. Z. Liang acknowledges support from the National Natural Science Foundation of China (Grant No. 11871036, 12271290). Y. Liu gratefully acknowledges financial support from the research startup fund at The Chinese University of Hong Kong, Shenzhen. 
The authors thank members of the group of mathematical finance at the Department of Mathematical Sciences, Tsinghua University for their feedback and useful conversations.

\appendix

\section{Further discussion on market models}\label{app:correlated}

In the classic literature (\cite{M1971,KLSX1991}), one uses the following standard Black-Scholes model to characterize the prices of $m \geq 1$ risky assets (stocks): for $i = 1, \cdots, m$, 
\begin{equation}\label{eq:incomplete}
	\d S_{i, t} = S_{i, t}\left(\mu_{i,t} \d t + \sum_{j=1}^l \sigma_{i,j,t} \d W_{j,t}\right), \;0 \leq t \leq T.
\end{equation}
Here, $\{S_{i, t}\}_{0 \leq t \leq T}$ is the price process for the $i$-th risky asset. They are modelled by a $l$-dimensional standard Brownian motion $\{\mathbf{W}_t = (W_{1,t}, \cdots, W_{l,t})',\; 0 \leq t \leq T\}$. Each one-dimensional Brownian motion $\{W_{j,t}\}_{0 \leq t \leq T}$, $j = 1, \cdots, l$, is interpreted as a source of randomness (systematic risks) and $l$ is the number of sources. It is always assumed more than the number of risky assets, i.e., $l \geq m$. 
Further, $\mu_{i,t}$ means the instantaneous return rate (or, appreciation rate) of the $i$-th risky asset at time $t$. For $1 \leq i \leq m$, $1 \leq j \leq l$, the volatility coefficient (or, dispersion coefficient) $\sigma_{i,j,t}$ models the instantaneous intensity how the $j$-th source of randomness affects the price of the $i$-th risky asset at time $t$. Usually, one includes a risk-free asset with the price process modelled by
\begin{equation}\label{eq:incomplete_risk-free}
	\d S_{0,t} = r_t S_{0,t} \d t,  \;0 \leq t \leq T,
\end{equation}
where $r_{t}$ means the interest rate at time $t$. 

Further, if $l = m$, the market is called complete, otherwise it is incomplete. In the incomplete market ($l > m$), one could not hedge all the randomness ($l$ sources) by constructing a portfolio with $m$ risky assets and one risk-free asset. 
As our focus in this paper is the PHARA utility, we start from the complete market with constant market parameters, which already covers a lot of essential features in the financial market. Under the setting of a complete market, we are able to derive the closed-form results and discuss financial insights. Hence, we use the standard Black-Scholes model in a complete market: for any $i = 1, \cdots, m$, 
\begin{equation}\label{eq:model_constant}
	\d S_{i, t} = S_{i, t}\left(\mu_{i} \d t + \sum_{j=1}^m \sigma_{i,j} \d W_{j,t}\right), \;0 \leq t \leq T,
\end{equation}
where $\mu_i$ is interpreted as the expected return rate of the $i$-th risky asset and $\sigma_{i, j}$ is interpreted as a constant volatility coefficient between the $i$-th risky asset and the Brownian motion $W_j$. This model is given by \eqref{eq:stock} in the main text. 
Hence, the matrix of volatility $\boldsymbol{\sigma} = (\sigma_{i, j})_{1 \leq i, j \leq m}$ is a square matrix. 
We assume that there exists some $\epsilon > 0$ such that 
\begin{equation}\label{eq:ass_posi_defi_app}
	\boldsymbol{\eta}^\top (\boldsymbol{\sigma} \boldsymbol{\sigma}^\top) \boldsymbol{\eta} \geq \epsilon ||\boldsymbol{\eta}||_2^2 \;\; \text{ for any } \boldsymbol{\eta} \in \R^m,
\end{equation}
where $||\boldsymbol{\eta}||_2 := \left(\sum_{i=1}^m \eta_i^2\right)^{\frac{1}{2}}$. 
This regularity condition \eqref{eq:ass_posi_defi_app} of $\boldsymbol{\sigma}$ is standard and widely used in the literature; see \cite{KLS1987} and \cite{C2000}. By Lemma 2.1 of \cite{KLS1987}, we have the matrices $\boldsymbol{\sigma}^\top$ and $\boldsymbol{\sigma}$ are invertible and hence $\boldsymbol{\Sigma} := \boldsymbol{\sigma} \boldsymbol{\sigma}^\top$ is also invertible. 
For a comprehensive discussion on financial models, we refer to Chapter 2 of \cite{KS1998}. 


In the standard Black-Scholes models \eqref{eq:incomplete} and \eqref{eq:model_constant}, it is implicit to see that the Brownian motions $W_1, \cdots, W_l$ may have correlations. For $j, k = 1, \cdots, l$, we denote by $\tau_{j, k}$ the correlation coefficient between two Brownian motions $W_{j}$ and $W_{k}$ and $\tau_{j, k}$ has the range $[-1, 1]$. 
We discuss this issue as follows. In some pricing models, the risky assets are characterized by the following model: for $i = 1, 2, \cdots, m$, 
\begin{equation}\label{eq:correlated}
	\d S_{i,t} = S_{i,t}(\mu_i \d t + \zeta_i \d Q_{i,t}), \;\; 0 \leq t \leq T,
\end{equation}
where the expected return rate $\mu_i$ is the same as the model \eqref{eq:model_constant}, $Q_{i}$ is a one-dimensional Brownian motion and $\zeta_i$ is the corresponding volatility parameter. In the model \eqref{eq:correlated}, each source of randomness is directly described by the risky asset itself, and all the randomness can hence be hedged by the portfolio constructed by risky assets. 

In Proposition \ref{prop:correlated}, we show that the model \eqref{eq:correlated} is a special case of the standard model \eqref{eq:incomplete} in a complete market, and the standard model \eqref{eq:model_constant} with independent Brownian motions is a special case of the model \eqref{eq:correlated}. These relationships help understand the modelling in financial practice.
\begin{proposition}\label{prop:correlated}
	\begin{enumerate}
		\item The model \eqref{eq:correlated} is a special case of the standard Black-Scholes model \eqref{eq:model_constant} in a complete market.
		
		\item 
		The standard Black-Scholes model \eqref{eq:model_constant} in a complete market with an $m$-dimensional independent standard Brownian motion is a special case of the model \eqref{eq:correlated}.
		
	\end{enumerate}
\end{proposition}

\begin{proof}
	\begin{enumerate}
		\item 
		Fix a setting in the model \eqref{eq:correlated}. We only need to specify the following parameters in the standard Black-Scholes model \eqref{eq:model_constant} in a complete market: for any $i,j = 1, \cdots, m$, 
		\begin{equation}
			W_i = Q_i, \;\; \sigma_{i,j} = \zeta_i \id_{\{i = j\}} + 0 \id_{\{ i\neq j \}}.
		\end{equation}
		Hence, the model \eqref{eq:correlated} becomes a special case of the model \eqref{eq:model_constant}.
		
		\item 
		Fix a setting in the standard Black-Scholes model \eqref{eq:model_constant} with an $m$-dimensional standard independent Brownian motion. We proceed to specify the following parameters in the model \eqref{eq:correlated}:
		for $i = 1, \cdots, m$, the $i$-th Brownian motion is given by
		\begin{equation}\label{eq:correlated_Q}
			Q_{i,t} := \frac{1}{\zeta_i} \sum_{j = 1}^m \sigma_{i, j} W_{j,t}, \;\; 0 \leq t \leq T,
		\end{equation}
		the volatility parameter of $\{Q_{i,t}\}_{0 \leq t \leq T}$ is given by 
		\begin{equation}\label{eq:correlated_vol}
			\zeta_i := \(\sum_{j=1}^m \sigma_{i, j}^2 \)^{\frac{1}{2}}, 
		\end{equation}
		and the correlation coefficient matrix is given by $(\rho_{i,k})_{1 \leq i, k \leq m}$, where
		\begin{equation}\label{eq:correlated_rho}
			\rho_{i,k} := 
			\frac{1}{\zeta_i \zeta_k} \sum_{j = 1}^m \sigma_{i,j}\sigma_{k, j},
			\; 1 \leq i, k \leq m.
		\end{equation}
		
		Here is the proof. As the $m$ Brownian motions $W_{1}, \cdots, W_{m}$ are independent, for any $i = 1, \cdots, m$, we have $Q_{i}$ given by \eqref{eq:correlated_Q} is a Brownian motion. Further, we have
		\begin{equation}
			\begin{aligned}
				\frac{\d S_{i, t}}{S_{i, t}} &= \mu_i \d t + \sum_{j = 1}^m \sigma_{i, j} \d W_{j,t}\\
				&= \mu_i \d t + \zeta_i \d Q_{i, t},
			\end{aligned}
		\end{equation}
		which is a special case of the model \eqref{eq:correlated}. 
		
		For $i, k = 1, \cdots, m$, for $t \in [0, T]$, we compute the covariance of $Q_{i}$ and $Q_{k}$ by \eqref{eq:correlated_Q}:
		\begin{equation}
			\cov\(Q_{i, t}, Q_{k, t}\) = \E\[Q_{i, t} Q_{k, t}\] = \frac{1}{\zeta_i \zeta_k} \E\[\(\sum_{j = 1}^m \sigma_{i, j} W_{j, t}\) \(\sum_{j = 1}^m \sigma_{k, j} W_{j, t}\)\] = \frac{t}{\zeta_i \zeta_l} \sum_{j=1}^m  \sigma_{i,j} \sigma_{k,j}.
		\end{equation}
		Hence, the correlation coefficient between $Q_{i}$ and $Q_{k}$ is given by \eqref{eq:correlated_rho}.
	\end{enumerate}
\end{proof}
In practical modelling, the standard Black-Scholes model \eqref{eq:incomplete} includes the most generality, but also has the most parameters to estimate. A setting of a complete market reduces some estimate cost, where the standard Black-Scholes model \eqref{eq:model_constant} enjoys the advantage of generality and the independence assumption between Brownian motions enjoys the ease of fewer estimation. The model \eqref{eq:correlated} is in between. One may choose among these different models according to the data availability and the statistical capability.



%
%
%
%

\section{Basic theory of concave envelope}\label{app:envelope}
Here we introduce the basic theory of the concave envelope, in terms of definition, derivation and application to portfolio selection. Interested readers may further refer to the theory of convex analysis (e.g., \cite{R1970} and \cite{HL2001}).


A direct definition of the concave envelope $U^{**}$ is defined in \eqref{eq:conv_dominating} from a graphical perspective. The superscript notation of $U^{**}$ comes from the perspective of convex analysis, where the concave envelope of a function $U$ can be also defined as:
\begin{equation}
	\label{eq_def_ConEnv}
	\left\{
	\begin{aligned}
		& U^{*}(y) := \sup\limits_{x \in \text{dom } U}  (U(x) - yx), \; && y > 0;\\
		& U^{**}(x) := \inf\limits_{y \in \text{dom } U^*} (U^{*}(y) + xy), \; && x \in \text{dom } U.
	\end{aligned}
	\right.
\end{equation}
Here we actually apply the Legendre-Fenchel transform with $f(x) := -U(-x)$:
\begin{equation}
	\left\{
	\begin{aligned}
		f^{*}(y) &:= \sup\limits_{x \in \text{dom } f}  (yx - f(x)), \; && y > 0;\\
		f^{**}(x) &:= \sup\limits_{y \in \text{dom } f^*}  (yx - f^*(x)), \; && x \in \text{dom f},
	\end{aligned}
	\right.
\end{equation}
where $f^*$ is known as the convex conjugate function and $f^{**}$ is known as the \textit{convex} envelope of $f$ or the convex biconjugate function. As we focus on the concave envelope, with an abuse of notation, we still denote by $U^{**}$ the \textit{concave} envelope. We further define the optimizer in $U^*$ by 
\begin{equation}
	I(y) := \arg \sup\limits_{x \in \text{dom } U} \{ U(x) - y x\}, \; y > 0.
\end{equation}

As indicated by the graphical definition  \eqref{eq:conv_dominating}, we can see that if $U$ is concave on $\D = \text{dom } U$, then $U^{**} = U$. To derive the concave envelope of a non-concave function $U$, we provide a procedure below and refer to a classic result Lemma \ref{lem:conc} for theoretical foundation. Lemma \ref{lem:conc} essentially originates from Lemma 2.3 of \cite{WXZ2019} and Lemma 5.1 of \cite{BC1994}, and is highly related to Lemma 6.3 of \cite{BS2014}. The proof is referred to Lemma 5.1 of \cite{BC1994}.
\begin{lemma}
	\label{lem:conc}
	Suppose that $U$ is continuous. The set $\mathcal{A} := \{ x \in \mathcal{D}: U(x) \neq U^{**}(x) \}$ is represented as a union of at most countably many disjoint open intervals, and $U^{**}$ is linear on each of the above intervals. 
\end{lemma}
Based on Lemma \ref{lem:conc}, if $U$ is not concave, we have a descriptive procedure to generate its concave envelope. We can use this procedure to establish the concave envelope in Figure \ref{fig}.
\begin{enumerate}[(i)]
	\item Find out the non-concave parts of $U$:
	$$
	\mathcal{B} := \left\{ x \in \D: U \text{ is not concave at } x \right\}.
	$$
	Write $\mathcal{B} = \cup_{j=1}^\infty \mathcal{B}_j$, where $\mathcal{B}_j$ is an interval. Note that the set $\mathcal{B}$ has plenty of common intervals as the set $\mathcal{A}$ in Lemma \ref{lem:conc}.
	
	\item On each $\mathcal{B}_j$, $j \geq 1$, use the tangent lines and linear connections which dominate $U$.
	
	\item Combine the concave parts and the linear parts as a new function $\tilde{U}$ dominating $U$. 
	
	\item Modify the function $\tilde{U}$ with new linear connections such that the modified function is concave on the domain. 
	
	\item Denote the modified function by $U^{**}$ and this function $U^{**}$ is the concave envelope of $U$.
\end{enumerate}

The concave envelope and the Legendre-Fenchel transform have an important application to portfolio selection in non-concave utility optimization. Assume the market is complete. According to the martingale method, Problem \eqref{prob-main} is equivalent to the terminal wealth optimization problem \eqref{prob-kernel}: 
\begin{equation}\label{prob-kernel}
	\begin{aligned}
		& \max_{X_T \in \F_T, \E[\xi_T X_T] \leq x_0, X_T \in \text{dom } U} \E[ U(X_T)].
	\end{aligned}
\end{equation}
Applying Lagrange duality methods and the above Equation \eqref{eq_def_ConEnv}, the optimal terminal wealth is given by
$$
X_T^* 
= \arg \sup\limits_{x \in \text{dom } U}  (U(x)- y \xi_T x) = I(y \xi_T),
$$
where the Lagrangian multiplier $y>0$ satisfies $\E[\xi_T I(y \xi_T)] = x_0$. For a complete procedure of solving the portfolio selection problem in non-concave utility optimization, we refer to Appendix A of \cite{LL2023}.


\section{Proofs}\label{app:proof}
\subsection{PHARA utility: Proof of Proposition \ref{prop:PHARA}}
\begin{proof}[Proof of Proposition \ref{prop:PHARA}]
	We prove the first statement. Let $\mathcal{A} = \{ x \in \D : U^{**}(x) - U(x) \neq 0 \}$. Then $\mathcal{A}$ is an open subset of $\D$ as the function $U^{**} - U$ is continuous.
	Write $\mathcal{A} = \cup_{i = 1}^\infty \mathcal{A}_i$, where $\mathcal{A}_i$ is an open interval. By Lemma \ref{lem:conc}, we have $U^{**}$ is linear and hence HARA on each $\mathcal{A}_i$, which implies $U^{**}$ is PHARA on the set $\mathcal{A}$. Because $U^{**} = U$ on $\D \setminus \mathcal{A}$, $U^{**}$ is PHARA on the set $\D \setminus \mathcal{A}$. As a whole, $U^{**}$ is PHARA on the domain $\D$.
	
	For the second statement, we show that there exists some $U$ which is not PHARA but $U^{**}$ is PHARA. We propose a case that $U$ is convex and not HARA on some interval $I$ and is concave and PHARA otherwise. On this interval $I$, we have $U^{**} \neq U$ and hence $I \subset \mathcal{A}$. By Lemma \ref{lem:conc},  $U^{**}$ is linear on the interval $I$ and hence $U^{**}$ is HARA. Thus, $U^{**}$ is HARA on each part of its domain. A particular example is as follow:
	\begin{equation}
		U(x) = \left\{
		\begin{aligned}
			& \sin (x) + 1, && \text{ if } x \in \left[\frac{3}{2} \pi, 2 \pi\right];\\
			& (x - 2 \pi)^{\frac{1}{2}} + 1, &&  \text{ if } x \in (2 \pi, \infty),
		\end{aligned}
		\right.
	\end{equation}
	and
	\begin{equation}
		U^{**}(x) = \left\{
		\begin{aligned}
			& \gamma_1 (x - \frac{3}{2} \pi), &&  \text{ if } x \in \left[\frac{3}{2} \pi, a_1\right];\\
			& (x - 2 \pi)^{\frac{1}{2}} + 1, &&  \text{ if } x \in (a_1, \infty),
		\end{aligned}
		\right.
	\end{equation}
	where the domains of $U$ and $U^{**}$ are the same $\left[\frac{3}{2}\pi, \infty\right)$, and $a_1 \in (2 \pi, \infty)$ is the unique solution of the tangent equation:
	\begin{equation}
		\frac{(a_1 - 2 \pi)^{\frac{1}{2}} + 1}{a_1 - \frac{3}{2} \pi} = \frac{1}{2} (a_1 - 2 \pi)^{-\frac{1}{2}},
	\end{equation}
	and we define its slope $\gamma_1 := \frac{(a_1 - 2 \pi)^{\frac{1}{2}} + 1}{a_1 - \frac{3}{2} \pi}$. Hence, $U^{**}$ is PHARA on $[\frac{3}{2} \pi, \infty)$ while $U$ is not.
\end{proof}

\subsection{Optimal wealth process: Proof of Proposition \ref{Thm_opt_wealth_process}}\label{sec:proof_wealth}

\begin{proof}[Proof of Proposition \ref{Thm_opt_wealth_process}]	
	(1) Based on the martingale and duality method, the optimal terminal wealth is obtained from $X_T^* = \arg\sup\limits_{x \in \mathcal D}[U(x) - y^* \xi_T x]$ satisfying \eqref{equation of nu*}. 
	According to Definition \ref{def_PHARA} and Equation \eqref{eq_def_U_tilde}, we solve for $k \in \{0,1,\cdots,n\}$:
	
	\noindent
	(i) 
	if $y^* \xi_T \in (\gamma_{k}^{+}, \gamma_k^{-})$, then
	\begin{equation}
		X_T^*=\arg\sup_{x \in \mathcal D} \left[U(x) - y^* \xi_T x\right] = a_k.
	\end{equation}
	
	\noindent
	(ii) 
	if $y^* \xi_T \in (\gamma_{k+1}^{-}, \gamma_k^{+})$, then
	\begin{equation}\label{eq:wealth1}
		\begin{aligned}
			X_T^* = \arg\sup_{x \in \mathcal D} \left[U(x) - y^* \xi_T x\right] 
			&= (U')^{-1}|_{(a_k, a_{k+1})}\left(y^* \xi_T\right)\\
			&= \left\{
			\begin{aligned}
				& A_k + \left(\frac{\gamma_k^{+}}{y^* \xi_T }\right)^{\!R_k^{-1}}\!\!(a_k - A_k), &&\text{ if } R_k \neq 0, \infty;\\ 
				& a_k + \frac{1}{\alpha_k} \log\left(\frac{\gamma_k^{+}}{y^* \xi_T}\right), &&\text{ if }  R_k = \infty, A_k = \infty, \alpha_k > 0.\\ 
			\end{aligned}
			\right.
		\end{aligned} 
	\end{equation}
	
	\noindent (iii)
	$\bullet$
	if $y^* \xi_T = \gamma^+_k$ where $\gamma_{k}^+ \neq \gamma_{k+1}^-$, we have 
	\begin{equation}\label{eq:wealth2}
		X_T^* = \arg\sup_{x \in \mathcal D} \left[U(x) - y^* \xi_T x\right] = a_k;	
	\end{equation}

	\noindent
	$\bullet$
	if $y^* \xi_T = \gamma^-_{k+1}$ where $\gamma_{k+1}^- \neq \gamma_{k}^+$, we have 
	\begin{equation}\label{eq:wealth3}
		X_T^* = \arg\sup_{x \in \mathcal D} \left[U(x) - y^* \xi_T x\right] = a_{k+1};
	\end{equation}
	
	\noindent
	$\bullet$ 
	if $y^* \xi_T = \gamma^+_k$ where $\gamma_{k}^+ = \gamma_{k+1}^-$ ($U^{**}$ is linear on $[a_k, a_{k+1}]$), we have $X_T^*$ is multi-valued, i.e., $X_T^*$ can be any real number in the interval $[a_k, a_{k+1}]$:
	\begin{equation}\label{eq:wealth4}
		X_T^* \in \arg\sup_{x \in \mathcal D} \left[U(x) - y^* \xi_T x\right]  = [a_k, a_{k+1}].
	\end{equation}
	We define the index set 
	\begin{equation}
		\mathcal{K} := \{ k \in 
		\{0, 1, \cdots, n\} | \gamma_k^+ = \gamma_{k+1}^- \}.
	\end{equation}
	But as $\xi_T$ has a continuous distribution (i.e., $\p(\xi_T = y) = 0$ for any $y \in \R$), we have
	\begin{equation}\label{eq:zero_prob}
		\begin{aligned}
			\p\left(y^* \xi_T \in \{ \gamma_k^+, \gamma_k^- \}_{k \in \mathcal{K}}\right) &\leq \p\left(y^* \xi_T \in \{ \gamma_k^+, \gamma_k^- \}_{k = 0}^{n} \cup \{\gamma_{n+1}^-\} \right)\\ 
			&= \sum_{k = 0}^{n} \left( \p( \xi_T = \gamma_k^{+}/y^*) + \p( \xi_T = \gamma_k^{-}/y^*)\right) +  \p( \xi_T = \gamma_{n+1}^-/y^*) = 0.
		\end{aligned}
	\end{equation}
	Thus, in the almost-sure sense, it does not matter what the value of $X_T^*$ is in this case. Moreover, $U^{**}(x) \neq U(x)$ if and only if $x \in \cup_{k \in \mathcal{K}}(a_k, a_{k+1})$. 
	Further, based on \eqref{eq:zero_prob}, we have
	\begin{equation}
		\begin{aligned}
			\E[U(X_T^*)] &= \E\left[U(X_T^*) \id_{\{X_T^* \in \cup_{k \in \mathcal{K}}(a_k, a_{k+1})\}}\right] + \E\left[U(X_T^*) \id_{\{X_T^* \in (\cup_{k \in \mathcal{K}}(a_k, a_{k+1}))^c\}}\right]\\
			&= \E\left[U(X_T^*) \id_{\cup_{k \in \mathcal{K}} \{y^* \xi_T = \gamma_k^+ \}}\right] + \E\left[U(X_T^*) \id_{\{X_T^* \in (\cup_{k \in \mathcal{K}}(a_k, a_{k+1}))^c\}}\right]\\
			&= 0+\E\left[U(X_T^*) \id_{\{X_T^* \in (\cup_{k \in \mathcal{K}}(a_k, a_{k+1}))^c\}}\right]\\
			&= 0+\E\left[U^{**}(X_T^*) \id_{\{X_T^* \in (\cup_{k \in \mathcal{K}}(a_k, a_{k+1}))^c\}}\right]\\
			&= \E\left[U^{**}(X_T^*) \id_{\{X_T^* \in \cup_{k \in \mathcal{K}}(a_k, a_{k+1})\}}\right] + \E\left[U^{**}(X_T^*) \id_{\{X_T^* \in (\cup_{k \in \mathcal{K}}(a_k, a_{k+1}))^c\}}\right]\\
			&= \E[U^{**}(X_T^*)].
		\end{aligned}
	\end{equation}
	Hence, we show that Problem \eqref{prob-main} is equivalent to the concavification problem \eqref{prob:concavification}, 
	i.e., Problems \eqref{prob-main} and \eqref{prob:concavification} have the same optimal portfolio $\boldsymbol{\pi}^*$ and the optimal terminal wealth $X_T^*$. Summarizing \eqref{eq:wealth1}-\eqref{eq:wealth4}, we have that the optimal wealth $X_T^*$ is given by \eqref{eq_terminal} and \eqref{equation of nu*}.
	
	\noindent
	(2) Above all, we give out facts that will be used in the following computation: for any $b > a > 0$,
	\begin{equation}\small
		\begin{aligned}
			&\quad \xi_t^{-1} \E[\xi_T \id_{\{y^* \xi_T \in (a, b) \}}| \mathcal{F}_t]
			= e^{-r(T-t)}\left[\Phi \left(d_1\left(\frac{a}{y^*\xi_t}\right)\right)-\Phi\left(d_1\left(\frac{b}{y^*\xi_t}\right)\right)\right],\\
			&\quad \xi_t^{-1} \E\left[\xi_T^{1-\frac{1}{R_k}} \id_{\{y^* \xi_T \in (a, b) \}}\Big| \mathcal{F}_t\right]\\
			&= e^{-r(T-t)}\left(\frac{y^*}{b}\right)^{\frac{1}{R_k}} \frac{\Phi'\left(d_1\left(\frac{b}{y^* \xi_t}\right)\right)}{\Phi'\left(d^k\left(\frac{b}{y^* \xi_t}\right)\right)}\left[\Phi\left(d^k\left(\frac{a}{y^* \xi_t}\right)\right)-\Phi\left(d^k\left(\frac{b}{y^* \xi_t}\right)\right)\right],\\
			&\quad \xi_t^{-1} \E\left[\xi_T \log(\xi_T) \id_{\{y^* \xi_T \in (a, b) \}}\Big| \mathcal{F}_t\right]\\
			&= e^{-r(T-t)} \left(\log (\xi_t)+ \left(-r+\frac{||\boldsymbol{\theta}||_2^2}{2}\right)(T-t) \right) \left[\Phi\left(d_1\left(\frac{a}{y^* \xi_t}\right)\right) - \Phi\left(d_1\left(\frac{b}{y^* \xi_t}\right)\right)\right]\\
			&\quad + e^{-r(T-t)} \left(||\boldsymbol{\theta}||_2 \sqrt{T-t}\right) \left[\Phi'\left(d_1\left(\frac{a}{y^* \xi_t}\right)\right) - \Phi'\left(d_1\left(\frac{b}{y^* \xi_t}\right)\right)\right].
		\end{aligned}
	\end{equation}
	According to the martingale representation argument and the expression \eqref{eq_terminal} of $X_T^*$, we compute: 
	\begin{equation}\small
		\begin{aligned}
			X_t^* &= \xi_t^{-1} \E\left[\xi_T X_T^* | \mathcal{F}_t\right]\\
			&=\sum_{k=0}^{n} \Bigg( e^{-r(T-t)} a_k \left[\Phi \left(d_1\left(\frac{\gamma_k^{+}}{y^*\xi_t}\right)\right)-\Phi\left(d_1\left(\frac{\gamma_k^{-}}{y^*\xi_t}\right)\right)\right]\\
			&\quad + e^{-r(T-t)}  A_k \left[\Phi \left(d_1\left(\frac{\gamma_{k+1}^{-}}{y^*\xi_t}\right)\right)-\Phi\left(d_1\left(\frac{\gamma_k^{+}}{y^*\xi_t}\right)\right)\right]  \id_{\{ R_k \neq 0,\infty \}} \\
			&\quad + e^{-r(T-t)} \left[ a_k+\frac{-||\boldsymbol{\theta}||_2 \sqrt{T-t}}{\alpha_k} d_1\left(\frac{\gamma_k^{+}}{y^* \xi_t}\right)\right] \times \\
			&\quad \quad \quad \quad \left[\Phi \left(d_1\left(\frac{\gamma_{k+1}^{-}}{y^*\xi_t}\right)\right)-\Phi\left(d_1\left(\frac{\gamma_k^{+}}{y^*\xi_t}\right)\right)\right] \id_{ \{ R_k = \infty, A_k =-\infty, \alpha_k>0 \} } \\  
			&\quad + e^{-r(T-t)} (a_k-A_k) \frac{\Phi'\left(d_1\left(\frac{\gamma_k^{+}}{y^* \xi_t}\right)\right)}{\Phi'\left(d^k\left(\frac{\gamma_k^{+}}{y^* \xi_t}\right)\right)} \left[\Phi \left(d^k\left(\frac{\gamma_{k+1}^{-}}{y^*\xi_t}\right)\right)-\Phi\left(d^k\left(\frac{\gamma_k^{+}}{y^*\xi_t}\right)\right)\right] \id_{\{ R_k \neq 0,\infty \}}\\
			&\quad + e^{-r(T-t)} \frac{-||\boldsymbol{\theta}||_2 \sqrt{T-t}}{\alpha_k} \left[\Phi' \left(d_1\left(\frac{\gamma_{k+1}^{-}}{y^*\xi_t}\right)\right)-\Phi'\left(d_1\left(\frac{\gamma_k^{+}}{y^*\xi_t}\right)\right)\right] \id_{ \{ R_k = \infty, A_k = -\infty, \alpha_k > 0 \} } \Bigg).
		\end{aligned}
	\end{equation}
\end{proof}

\subsection{Optimal portfolio (general case): Proof of Theorem \ref{thm_opt_strategy}}

\begin{proof}[Proof of Theorem \ref{thm_opt_strategy}]
	According to the martingale method, we apply It\^{o}'s formula to $X_t = X^*_t(\xi_t)$ by using \eqref{eq:xi} and \eqref{sde of x}, and obtain the optimal portfolio vector
	\begin{equation}
		\boldsymbol{\pi}_t^* = -(\boldsymbol{\sigma}^\top)^{-1} \boldsymbol{\theta} \xi_t \frac{\partial X_t^*(\xi_t)}{\partial\xi_t}.
	\end{equation}
	In the following computation, we use facts that
	\begin{equation}\small
		\begin{aligned}
			&(-\xi_t)\cdot \frac{\partial}{\partial \xi_t}d_1\left(\frac{\gamma_{k+1}^{-}}{y^* \xi_t}\right)=-\frac{1}{||\boldsymbol{\theta}||_2 \sqrt{T-t}},\\
			&(-\xi_t)\cdot \frac{\partial}{\partial \xi_t}\Phi\left(d_1\left(\frac{\gamma_{k+1}^{-}}{y^* \xi_t}\right)\right)=-\frac{1}{||\boldsymbol{\theta}||_2 \sqrt{T-t}}\Phi'\left(d_1\left(\frac{\gamma_{k+1}^{-}}{y^* \xi_t}\right)\right),\\
			&(-\xi_t)\cdot \frac{\partial}{\partial \xi_t}\Phi'\left(d_1\left(\frac{\gamma_{k+1}^{-}}{y^* \xi_t}\right)\right)
			= -\frac{1}{||\boldsymbol{\theta}||_2 \sqrt{T-t}}\Phi'\left(d_1\left(\frac{\gamma_{k+1}^{-}}{y^* \xi_t}\right)\right)\left(-d_1\left(\frac{\gamma_{k+1}^{-}}{y^* \xi_t}\right)\right).
		\end{aligned}
	\end{equation}
	For each $k \in \{0, \cdots, n\}$, we compute as follows:
	
	\noindent
	$\bullet$
	if $R_k = 0$, we have $\gamma_k^{+}=\gamma_{k+1}^{-}$, and hence
	\begin{equation}
		\begin{aligned}
			&\quad \quad (-\xi_t) \frac{\partial }{\partial\xi_t} \left\{a_k   \left[\Phi\left(d_1\left(\frac{\gamma_k^{+}}{y^*\xi_t}\right)\right) - \Phi\left(d_1\left(\frac{\gamma_k^{-}}{y^*\xi_t}\right)\right)\right]\right\}\\
			&=\frac{a_k}{-||\boldsymbol{\theta}||_2\sqrt{T-t}}\left[\Phi'\left(d_1\left(\frac{\gamma_k^{+}}{y^*\xi_t}\right)\right) - \Phi'\left(d_1\left(\frac{\gamma_k^{-}}{y^*\xi_t}\right)\right)\right]\\
			&=\frac{a_k-a_{k+1}}{-||\boldsymbol{\theta}||_2\sqrt{T-t}}\Phi'\left(d_1\left(\frac{\gamma_{k+1}^{-}}{y^*\xi_t}\right)\right) +\frac{a_{k+1}}{-||\boldsymbol{\theta}||_2\sqrt{T-t}}\Phi'\left(d_1\left(\frac{\gamma_{k+1}^{-}}{y^*\xi_t}\right)\right) - \frac{a_k}{-||\boldsymbol{\theta}||_2\sqrt{T-t}}\Phi'\left(d_1\left(\frac{\gamma_k^{-}}{y^*\xi_t}\right)\right);
		\end{aligned}
	\end{equation}
	
	\noindent
	$\bullet$
	if $R_k=\infty,~ A_k =-\infty,~ \alpha_k > 0$, we have
	\begin{equation}
		\small
		\begin{aligned}
			&\quad \quad (-\xi_t) \frac{\partial }{\partial\xi_t} 
			\Bigg\{
			a_k   \left[\Phi\left(d_1\left(\frac{\gamma_k^{+}}{y^*\xi_t}\right)\right) - \Phi\left(d_1\left(\frac{\gamma_k^{-}}{y^*\xi_t}\right)\right)\right]\\
			&\quad \quad \quad \quad \quad\quad\quad+ \left(a_k + \frac{-||\boldsymbol{\theta}||_2 \sqrt{T-t}}{\alpha_k} d_1\left(\frac{\gamma_k^{+}}{y^* \xi_t}\right)\right) \left[\Phi\left(d_1\left(\frac{\gamma_{k+1}^{-}}{y^*\xi_t}\right)\right) - \Phi\left(d_1\left(\frac{\gamma_k^{+}}{y^*\xi_t}\right)\right)\right]\\
			&\quad \quad \quad \quad \quad\quad\quad+ \frac{-||\boldsymbol{\theta}||_2 \sqrt{T-t}}{\alpha_k} \left[\Phi'\left(d_1\left(\frac{\gamma_{k+1}^{-}}{y^*\xi_t}\right)\right) - \Phi'\left(d_1\left(\frac{\gamma_k^{+}}{y^*\xi_t}\right)\right)\right]
			\Bigg\}\\
			&= \frac{a_k}{-||\boldsymbol{\theta}||_2 \sqrt{T-t}} \left[\Phi'\left(d_1\left(\frac{\gamma_{k+1}^{-}}{y^*\xi_t}\right)\right) - \Phi'\left(d_1\left(\frac{\gamma_k^{-}}{y^*\xi_t}\right)\right)\right]\\
			&\quad \quad + \frac{1}{\alpha_k} \left[\Phi\left(d_1\left(\frac{\gamma_{k+1}^{-}}{y^*\xi_t}\right)\right) - \Phi\left(d_1\left(\frac{\gamma_k^{+}}{y^*\xi_t}\right)\right)\right]+\frac{-1}{\alpha_k}\left(d_1\left(\frac{\gamma_{k+1}^{-}}{y^* \xi_t}\right) - d_1\left(\frac{\gamma_{k}^{+}}{y^* \xi_t}\right)\right) \Phi'\left(d_1\left(\frac{\gamma_{k+1}^{-}}{y^* \xi_t}\right)\right)\\
			&= \frac{a_k}{-||\boldsymbol{\theta}||_2 \sqrt{T-t}} \left[\Phi'\left(d_1\left(\frac{\gamma_{k+1}^{-}}{y^*\xi_t}\right)\right) - \Phi'\left(d_1\left(\frac{\gamma_k^{-}}{y^*\xi_t}\right)\right)\right]\\
			&\quad \quad + \frac{1}{\alpha_k} \left[\Phi\left(d_1\left(\frac{\gamma_{k+1}^{-}}{y^*\xi_t}\right)\right) - \Phi\left(d_1\left(\frac{\gamma_k^{+}}{y^*\xi_t}\right)\right)\right] +\frac{1}{-||\boldsymbol{\theta}||_2 \sqrt{T-t}}(a_{k+1}-a_k) \Phi'\left(d_1\left(\frac{\gamma_{k+1}^{-}}{y^* \xi_t}\right)\right)\\
			&=\frac{1}{\alpha_k} \left[\Phi\left(d_1\left(\frac{\gamma_{k+1}^{-}}{y^*\xi_t}\right)\right) - \Phi\left(d_1\left(\frac{\gamma_k^{+}}{y^*\xi_t}\right)\right)\right]+\frac{a_{k+1}}{-||\boldsymbol{\theta}||_2\sqrt{T-t}}\Phi'\left(d_1\left(\frac{\gamma_{k+1}^{-}}{y^*\xi_t}\right)\right) - \frac{a_k}{-||\boldsymbol{\theta}||_2\sqrt{T-t}}\Phi'\left(d_1\left(\frac{\gamma_k^{-}}{y^*\xi_t}\right)\right),
		\end{aligned}
	\end{equation}
	where in the second equality we use $\gamma_{k+1}^{-} = \gamma_{k}^{+} e^{-\alpha_k (a_{k+1}-a_k)}$ (due to \eqref{eq_def_U_tilde});
	
	\noindent
	$\bullet$
	if $R_k \neq 0, \infty$, we have
	\begin{equation}\small
		\begin{aligned}
			&\quad \quad(-\xi_t) \frac{\partial }{\partial\xi_t} 
			\Bigg\{ a_k  \left[\Phi\left(d_1\left(\frac{\gamma^+_k}{y^*\xi_t}\right)\right) -\Phi\left(d_1\left(\frac{\gamma^-_k}{y^* \xi_t}\right)\right)\right] 
			+ A_k \left[ \Phi\left(d_1\left(\frac{\gamma^-_{k+1}}{y^*\xi_t}\right)\right) -\Phi\left(d_1\left(\frac{\gamma^+_k}{y^*\xi_t}\right)\right)\right] \\
			&\quad \quad \quad \quad \quad+ \left(a_k-A_k\right) \frac{\Phi'\left(d_1\left(\frac{\gamma^+_k}{y^*\xi_t}\right)\right)}{\Phi'\left(d^k\left(\frac{\gamma^+_k}{y^*\xi_t}\right)\right)}\left[ \Phi\left(d^k\left(\frac{\gamma^-_{k+1}}{y^*\xi_t}\right)\right) -\Phi\left(d^k\left(\frac{\gamma^+_k}{y^*\xi_t}\right)\right)\right] \Bigg\}\\
			& = \frac{a_k}{-||\boldsymbol{\theta}||_2 \sqrt{T-t}} \left[\Phi'\left(d_1\left(\frac{\gamma^+_k}{y^*\xi_t}\right)\right) -\Phi'\left(d_1\left(\frac{\gamma^-_k}{y^* \xi_t}\right)\right)\right] + \frac{A_k}{-||\boldsymbol{\theta}||_2 \sqrt{T-t}} \left[ \Phi'\left(d_1\left(\frac{\gamma^-_{k+1}}{y^*\xi_t}\right)\right) -\Phi'\left(d_1\left(\frac{\gamma^+_k}{y^*\xi_t}\right)\right)\right]\\
			& \quad \quad+\frac{ a_k-A_k}{-||\boldsymbol{\theta}||_2 \sqrt{T-t}} \frac{\Phi'\left(d_1\left(\frac{\gamma_k^{+}}{y^*\xi_t}\right)\right)}{\Phi'\left(d^k\left(\frac{\gamma_k^{+}}{y^*\xi_t}\right)\right)}\left[ \Phi'\left(d^k\left(\frac{\gamma_{k+1}^{-}}{y^*\xi_t}\right)\right) -\Phi'\left(d^k\left(\frac{\gamma_{k}^{+}}{y^*\xi_t}\right)\right) \right] +R_k^{-1} X^R_{t,k}e^{r(T-t)}\\
			&= \frac{a_k}{-||\boldsymbol{\theta}||_2 \sqrt{T-t}} \left[\Phi'\left(d_1\left(\frac{\gamma^+_k}{y^*\xi_t}\right)\right) -\Phi'\left(d_1\left(\frac{\gamma^-_k}{y^* \xi_t}\right)\right)\right] + \frac{A_k}{-||\boldsymbol{\theta}||_2 \sqrt{T-t}} \left[ \Phi'\left(d_1\left(\frac{\gamma^-_{k+1}}{y^*\xi_t}\right)\right) -\Phi'\left(d_1\left(\frac{\gamma^+_k}{y^*\xi_t}\right)\right)\right]\\
			&\quad \quad+\frac{ a_k-A_k}{||\boldsymbol{\theta}||_2 \sqrt{T-t}} \Phi'\left(d_1\left(\frac{\gamma_k^{+}}{y^*\xi_t}\right)\right)+ 
			\frac{ a_{k+1}-A_k}{-||\boldsymbol{\theta}||_2 \sqrt{T-t}} \Phi'\left(d_1\left(\frac{\gamma_{k+1}^{-}}{y^*\xi_t}\right)\right) +R_k^{-1} X^R_{t,k}e^{r(T-t)}\\
			&=R_k^{-1} X^R_{t,k}e^{r(T-t)}
			+\frac{a_{k+1}}{-||\boldsymbol{\theta}||_2\sqrt{T-t}}\Phi'\left(d_1\left(\frac{\gamma_{k+1}^{-}}{y^*\xi_t}\right)\right) - \frac{a_k}{-||\boldsymbol{\theta}||_2\sqrt{T-t}}\Phi'\left(d_1\left(\frac{\gamma_k^{-}}{y^*\xi_t}\right)\right),
		\end{aligned}
	\end{equation}
	where the second equality holds because $\frac{\Phi'(d_1 (\frac{U'(x)}{y^*\xi_t}))}{\Phi'(d^k(\frac{U'(x)}{y^*\xi_t}))} (x-A_k)$ is constant over $(a_k, a_{k+1})$; especially, 
	$
	\frac{\Phi'(d_1 (\frac{\gamma_{k+1}^{-}}{y^*\xi_t} ))}{\Phi'(d^k (\frac{\gamma_{k+1}^{-}}{y^*\xi_t}))}(a_{k+1}-A_k)  =\frac{\Phi'(d_1 (\frac{\gamma_k^{+}}{y^*\xi_t}))}{\Phi' (d^k (\frac{\gamma_k^{+}}{y^*\xi_t}))} (a_{k}-A_k).
	$
	Finally, adding up each term, we derive the optimal portfolio 
	vector:
	\begin{equation}\label{eq:port}
		\begin{aligned}
			\boldsymbol{\pi}^*_t &= -(\boldsymbol{\sigma}^\top)^{-1} \boldsymbol{\theta} \xi_t \frac{\partial X_t^*}{\partial\xi_t}\\
			&= (\boldsymbol{\sigma}^\top)^{-1} \boldsymbol{\theta} \sum_{k=0}^n \left(-\xi_t\right) \frac{\partial }{\partial\xi_t} \left( X^D_{t,k} + X^A_{t,k} + X^{\bar{A}}_{t,k} + X^R_{t,k} + X^{\bar{R}}_{t,k} \right)\\
			&= (\boldsymbol{\sigma}^\top)^{-1} \boldsymbol{\theta} \sum_{k=0}^n \bigg\{ R_k^{-1} X^R_{t,k}  + e^{-r(T-t)}\frac{a_k-a_{k+1}}{-||\boldsymbol{\theta}||_2\sqrt{T-t}}\Phi'\left(d_1\left(\frac{\gamma_{k+1}^{-}}{y^*\xi_t}\right)\right) \id_{\{ R_k = 0 \}} \\
			&\quad \quad \quad \quad \quad+ e^{-r(T-t)}\frac{1}{\alpha_k} \left[\Phi\left(d_1\left(\frac{\gamma_{k+1}^{-}}{y^*\xi_t}\right)\right) - \Phi\left(d_1\left(\frac{\gamma_k^{+}}{y^*\xi_t}\right)\right)\right] \id_{\{ R_k = \infty, A_k = -\infty, \alpha_k > 0 \}} \bigg\},
		\end{aligned}
	\end{equation}
	where we use facts that $R_n \neq 0$, $\gamma_{n+1}^{-}=0$ and $\gamma_0^{-}=\infty$.
\end{proof}

\subsection{Optimal portfolio: Proof of Theorem \ref{thm_main}}

\begin{proof}[Proof of Theorem \ref{thm_main}]
	If $R_k \in \{R, 0\}$, then the optimal portfolio vector \eqref{eq:main} is derived from the general case \eqref{eq:port} in the proof of Theorem \ref{thm_opt_strategy} by
	\begin{equation}
		\begin{aligned}
			\boldsymbol{\pi}^*_t &= (\boldsymbol{\sigma}^\top)^{-1} \boldsymbol{\theta} \sum_{k=0}^n \bigg\{ R^{-1} X^R_{t,k} \id_{\{ R_k \neq 0,\infty \}} + e^{-r(T-t)}\frac{a_k-a_{k+1}}{-||\boldsymbol{\theta}||_2\sqrt{T-t}}\Phi'\left(d_1\left(\frac{\gamma_{k+1}^{-}}{y^*\xi_t}\right)\right) \id_{\{ R_k = 0 \}} \bigg\}\\
			&= \frac{(\boldsymbol{\sigma}^\top)^{-1} \boldsymbol{\theta}}{R} \sum_{k=0}^n X^R_{t,k} \id_{\{ R_k \neq 0 \}} + \frac{(\boldsymbol{\sigma}^\top)^{-1} \boldsymbol{\theta}}{||\boldsymbol{\theta}||_2} \frac{e^{-r(T-t)} }{\sqrt{T-t}} \sum_{k=0}^n  (a_{k+1} - a_k) \Phi'\left(d_1\left(\frac{\gamma_{k+1}^{-}}{y^*\xi_t}\right)\right) \id_{\{ R_k = 0 \}} \\
			&= \frac{(\boldsymbol{\sigma}^\top)^{-1} \boldsymbol{\theta}}{R} \left(X_t^* - \sum_{k=0}^n X^A_{t,k} - \sum_{k=0}^n X^D_{t,k} \right) + \frac{(\boldsymbol{\sigma}^\top)^{-1} \boldsymbol{\theta}}{||\boldsymbol{\theta}||_2} \frac{e^{-r(T-t)} }{\sqrt{T-t}} \sum_{k=0}^n  (a_{k+1} - a_k) \Phi'\left(d_1\left(\frac{\gamma_{k+1}^{-}}{y^*\xi_t}\right)\right) \id_{\{ R_k = 0 \}} \\
			&= \frac{(\boldsymbol{\sigma}^\top)^{-1} \boldsymbol{\theta}}{R} X_t^* + \frac{(\boldsymbol{\sigma}^\top)^{-1} \boldsymbol{\theta}}{||\boldsymbol{\theta}||_2}\frac{e^{-r(T-t)} }{\sqrt{T-t}} \sum_{k=0}^n  (a_{k+1} - a_k) \Phi'\left(d_1\left(\frac{\gamma_{k+1}^{-}}{y^*\xi_t}\right)\right) \id_{\{ R_k = 0 \}} \\
			&\quad \quad - \frac{(\boldsymbol{\sigma}^\top)^{-1} \boldsymbol{\theta}}{R} e^{-r(T-t)} \sum_{k=0}^n A_k \left[\Phi\left(d_1\left(\frac{\gamma_{k+1}^-}{y^* \xi_t}\right)\right) - \Phi\left(d_1\left(\frac{\gamma_k^+}{y^* \xi_t}\right)\right)\right] \id_{\{ R_k \neq 0\} } \\
			&\quad\quad - \frac{(\boldsymbol{\sigma}^\top)^{-1} \boldsymbol{\theta}}{R} e^{-r(T-t)} \sum_{k=0}^n a_k \left[\Phi\left(d_1\left(\frac{\gamma_{k}^+}{y^* \xi_t}\right)\right) - \Phi\left(d_1\left(\frac{\gamma_k^-}{y^* \xi_t}\right)\right)\right].
		\end{aligned}
	\end{equation}
\end{proof}

\subsection{Asymptotic analysis: Proof of Theorem \ref{thm_option table}}
\begin{proof}[Proof of Theorem \ref{thm_option table}]
	As $R_k \in \{R, 0\}$, according to Proposition \ref{Thm_opt_wealth_process}, we have 
	\begin{equation}
		X_t^* = \sum_{k = 0}^n \( X_{t,k}^D + X_{t,k}^A + X_{t,k}^R \).
	\end{equation}
	Now we fix $t \in (0, T)$ and $y^* > 0$ and conduct an asymptotic analysis. 
	\begin{enumerate}[(i)]
		\item 
		Using \eqref{eq:d1}, for any $\gamma \in (0, \infty]$, as $\xi_t \rightarrow 0$, we have 
		\begin{equation}\label{eq:d1xi0}
			\frac{\gamma}{y^* \xi_t} \rightarrow \infty, \;\; d_1\(\frac{\gamma}{y^* \xi_t}\) \rightarrow -\infty, 
		\end{equation}
		and
		\begin{equation}\label{eq:phixi0}
			\Phi'\(d_1\(\frac{\gamma}{y^* \xi_t}\)\) \rightarrow 0, \;\; \Phi\(d_1\(\frac{\gamma}{y^* \xi_t}\)\) \rightarrow 0. 
		\end{equation}  
		Hence, for $\gamma \in \{\gamma_0^-, \gamma_0^+, \gamma_1^-, \gamma_1^+, \cdots, \gamma_n^-, \gamma_n^+\}$, we have \eqref{eq:d1xi0}-\eqref{eq:phixi0}. 
		For $\gamma = \gamma_{n+1}^- = 0$, for any $\xi_t \in (0, \infty)$, we have
		\begin{equation}
			\frac{\gamma_{n+1}^-}{y^* \xi_t} = 0, \;\; d_1\(\frac{\gamma_{n+1}^-}{y^* \xi_t}\) = \infty,  
		\end{equation}
		and
		\begin{equation}
			\Phi'\(d_1\(\frac{\gamma_{n+1}^-}{y^* \xi_t}\)\) = 0, \;\; \Phi\(d_1\(\frac{\gamma_{n+1}^-}{y^* \xi_t}\)\) = 1. 
		\end{equation} 
		Now we compute: for any $\gamma > 0$, 
		\begin{equation}\label{eq:phi'ratio}
			\frac{\Phi'\( d_1\(\frac{\gamma}{y^* \xi_t}\) \)}{\Phi'\( d^k\(\frac{\gamma}{y^* \xi_t}\) \)} = \exp\left\{(T-t) \( \frac{r}{R_k} + \frac{||\boldsymbol{\theta}||_2^2}{2 R_k^2}(1-R_k) \)\right\} \(\frac{y^* \xi_t}{\gamma}\)^{-\frac{1}{R_k}}.
		\end{equation}    
		Based on the expression \eqref{eq:wealth_five_term} of $X_t^*$, as $\xi_t \rightarrow 0$, for any $k \in \{0, 1, \cdots, n\}$, we have
		\begin{equation}
			X^D_{t,k} = e^{-r(T-t)} a_k  \left[ \Phi\left(d_1\left(\frac{\gamma^+_k}{y^*\xi_t}\right)\right) -\Phi\left(d_1\left(\frac{\gamma^-_k}{y^*\xi_t}\right)\right) \right] \rightarrow 0, 
		\end{equation}	
		\begin{equation}
			\begin{aligned}
				X^A_{t,k} &=  e^{-r(T-t)} A_k \left[ \Phi\left(d_1\left(\frac{\gamma^-_{k+1}}{y^*\xi_t}\right)\right) -\Phi\left(d_1\left(\frac{\gamma^+_k}{y^*\xi_t}\right)\right) \right] \id_{\left\{ R_k \neq 0,\infty \right\}}\\
				&\rightarrow \left\{
				\begin{aligned}
					& e^{-r(T-t)} A_n, && \text{ if } k = n;\\ 
					& 0, && \text{ if } k \neq n.
				\end{aligned}
				\right.
			\end{aligned}
		\end{equation}
		For any $k \in \{0, 1, \cdots, n-1\}$ with $R_k \neq 0$, by the expression of the PHARA utility on $[a_k, a_{k+1}]$, $(a_k-A_k)^{1-R_k}$ is well-defined and we have $a_k > A_k$. We also have $\gamma_k^+ > \gamma_{k+1}^-$ by strict concavity, which means $\Phi\left(d^k\left(\frac{\gamma^-_{k+1}}{y^*\xi_t}\right)\right) -\Phi\left(d^k\left(\frac{\gamma^+_k}{y^*\xi_t}\right)\right) > 0$. Hence, as $\xi_t \rightarrow 0$, we have  
		\begin{equation}\label{eq:X^R_tk}
			X^R_{t,k} =  e^{-r(T-t)} (a_k-A_k) \frac{\Phi'\left(d_1\left(\frac{\gamma^+_k}{y^*\xi_t}\right)\right)}{\Phi'\left(d^k\left(\frac{\gamma^+_k}{y^*\xi_t}\right)\right)}\left[ \Phi\left(d^k\left(\frac{\gamma^-_{k+1}}{y^*\xi_t}\right)\right) -\Phi\left(d^k\left(\frac{\gamma^+_k}{y^*\xi_t}\right)\right) \right] \id_{\left\{ R_k \neq 0,\infty \right\}} \geq 0.
		\end{equation}
		For $k=n$, according to the expression of \eqref{eq:phi'ratio}, as $\xi_t \rightarrow 0$, we have 
		\begin{equation}
			\frac{\Phi'\( d_1\(\frac{\gamma_n^+}{y^* \xi_t}\) \)}{\Phi'\( d^k\(\frac{\gamma_n^+}{y^* \xi_t}\) \)} \rightarrow \infty, 
		\end{equation}
		and
		\begin{equation}
			\Phi\left(d^k\left(\frac{\gamma^-_{n+1}}{y^*\xi_t}\right)\right) -\Phi\left(d^k\left(\frac{\gamma^+_n}{y^*\xi_t}\right)\right) = 
			1 -\Phi\left(d^k\left(\frac{\gamma^+_n}{y^*\xi_t}\right)\right) \rightarrow 1,
		\end{equation}
		which implies
		\begin{equation}\label{eq:X^R_tn}
			X^R_{t,n} =  e^{-r(T-t)} (a_n-A_n) \frac{\Phi'\left(d_1\left(\frac{\gamma^+_n}{y^*\xi_t}\right)\right)}{\Phi'\left(d^{n}\left(\frac{\gamma^+_n}{y^*\xi_t}\right)\right)}\left[ \Phi\left(d^{n}\left(\frac{\gamma^-_{n+1}}{y^*\xi_t}\right)\right) -\Phi\left(d^{n}\left(\frac{\gamma^+_n}{y^*\xi_t}\right)\right) \right] \id_{\left\{ R_n \neq 0,\infty \right\}} \rightarrow \infty.
		\end{equation}
		Hence, combining \eqref{eq:X^R_tk} and \eqref{eq:X^R_tn}, we have
		\begin{equation}
			X^R_t = \sum_{k=0}^n X_{t,k}^R \rightarrow \infty,
		\end{equation}
		and
		\begin{equation}
			X_t^* \rightarrow \infty.
		\end{equation}
		Moreover, according to the expressions of \eqref{eq:pkqk}, for any $k \in \{0, \cdots, n-1\}$, as $\xi_t \rightarrow 0$, we have 
		\begin{equation}
			p_k \rightarrow 0, \; q_k \rightarrow 0, 
		\end{equation}
		and
		\begin{equation}
			p_n \rightarrow 0, \; q_n \rightarrow 1, 
		\end{equation}
		where we use $\gamma_{n+1}^- = 0$. Thus, we obtain that if $\xi_t \rightarrow 0$, then
		\begin{equation}
			\boldsymbol{\pi}_t^{(2)} \rightarrow \mathbf{0},\;\; \boldsymbol{\pi}_t^{(3)} \rightarrow -\frac{A_n e^{-r(T-t)}}{R} (\boldsymbol{\sigma}^\top)^{-1} \boldsymbol{\theta},\;\; \frac{\boldsymbol{\pi}_t^*}{X_t^*} \rightarrow \frac{1}{R} (\boldsymbol{\sigma}^\top)^{-1} \boldsymbol{\theta}, \;\; \boldsymbol{\pi}_t \rightarrow \boldsymbol{\infty}.
		\end{equation}
		
		\item 
		An important point throughout this part of proof is that (a) $\gamma_0^- = \infty$ always holds; and (b) $\gamma_0^+ = \infty$ if and only if $a_0 = A_0$ and $R_0 \neq 0$. We can show (a)-(b) by Definition \ref{def_PHARA} of the PHARA utility.
		
		Using \eqref{eq:d1}, for any $\gamma \in [0, \infty)$, as $\xi_t \rightarrow \infty$, we have 
		\begin{equation}\label{eq:d1xi_inf}
			\frac{\gamma}{y^* \xi_t} \rightarrow 0, \;\; d_1\(\frac{\gamma}{y^* \xi_t}\) \rightarrow \infty, 
		\end{equation}
		and
		\begin{equation}\label{eq:phixi_inf}
			\Phi'\(d_1\(\frac{\gamma}{y^* \xi_t}\)\) \rightarrow 0, \;\; \Phi\(d_1\(\frac{\gamma}{y^* \xi_t}\)\) \rightarrow 1. 
		\end{equation}  
		Hence, for $\gamma \in \{\gamma_1^-, \gamma_1^+, \cdots, \gamma_n^-, \gamma_n^+, \gamma_{n+1}^-\}$, we have \eqref{eq:d1xi_inf}-\eqref{eq:phixi_inf}. 
		For $\gamma = \infty$, for any $\xi_t \in (0, \infty)$, we have
		\begin{equation}
			\frac{\infty}{y^* \xi_t} = \infty, \;\; d_1\(\frac{\infty}{y^* \xi_t}\) = -\infty,  
		\end{equation}
		\begin{equation}
			\Phi'\(d_1\(\frac{\infty}{y^* \xi_t}\)\) = 0, \;\; \Phi\(d_1\(\frac{\infty}{y^* \xi_t}\)\) = 0. 
		\end{equation} 
		For any $\gamma \in [0, \infty)$, as $\xi_t \rightarrow \infty$,
		\begin{equation}\label{eq:phi'ratio_xiinf}
			\frac{\Phi'\( d_1\(\frac{\gamma}{y^* \xi_t}\) \)}{\Phi'\( d^k\(\frac{\gamma}{y^* \xi_t}\) \)} = \exp\left\{(T-t) \( \frac{r}{R_k} + \frac{||\boldsymbol{\theta}||_2^2}{2 R_k^2}(1-R_k) \)\right\} \(y^* \xi_t\)^{-\frac{1}{R_k}} \gamma^{\frac{1}{R_k}} \rightarrow 0.
		\end{equation}    
		Based on the expression \eqref{eq:wealth_five_term} of $X_t^*$, as $\xi_t \rightarrow \infty$, for any $k \in \{1, \cdots, n\}$, we have
		\begin{equation}
			\begin{aligned}
				X^D_{t,k} &= e^{-r(T-t)} a_k  \left[ \Phi\left(d_1\left(\frac{\gamma^+_k}{y^*\xi_t}\right)\right) -\Phi\left(d_1\left(\frac{\gamma^-_k}{y^*\xi_t}\right)\right) \right] 
				\rightarrow 0,
			\end{aligned}
		\end{equation}	
		\begin{equation}
			\begin{aligned}
				X^A_{t,k} &=  e^{-r(T-t)} A_k \left[ \Phi\left(d_1\left(\frac{\gamma^-_{k+1}}{y^*\xi_t}\right)\right) -\Phi\left(d_1\left(\frac{\gamma^+_k}{y^*\xi_t}\right)\right) \right] \id_{\left\{ R_k \neq 0,\infty \right\}} \rightarrow 0,
			\end{aligned}
		\end{equation}
		\begin{equation}\label{eq:X^R_tk_xiinf}
			X^R_{t,k} =  e^{-r(T-t)} (a_k-A_k) \frac{\Phi'\left(d_1\left(\frac{\gamma^+_k}{y^*\xi_t}\right)\right)}{\Phi'\left(d^k\left(\frac{\gamma^+_k}{y^*\xi_t}\right)\right)}\left[ \Phi\left(d^k\left(\frac{\gamma^-_{k+1}}{y^*\xi_t}\right)\right) -\Phi\left(d^k\left(\frac{\gamma^+_k}{y^*\xi_t}\right)\right) \right] \id_{\left\{ R_k \neq 0,\infty \right\}} \rightarrow 0.
		\end{equation}
		
		For $k = 0$, if $\gamma_0^+ \in [0,\infty)$, then $a_0 > A_0$ or $R_0 = 0$. Because $\gamma_0^- = \infty$, we have
		\begin{equation}
			\begin{aligned}
				X^D_{t,0} &= e^{-r(T-t)} a_0  \left[ \Phi\left(d_1\left(\frac{\gamma^+_0}{y^*\xi_t}\right)\right) -\Phi\left(d_1\left(\frac{\gamma^-_0}{y^*\xi_t}\right)\right) \right] 
				\rightarrow e^{-r(T-t)} a_0,
			\end{aligned}
		\end{equation}	
		\begin{equation}
			\begin{aligned}
				X^A_{t,0} &=  e^{-r(T-t)} A_0 \left[ \Phi\left(d_1\left(\frac{\gamma^-_1}{y^*\xi_t}\right)\right) -\Phi\left(d_1\left(\frac{\gamma^+_0}{y^*\xi_t}\right)\right) \right] \id_{\left\{ R_0 \neq 0,\infty \right\}} \rightarrow 0,
			\end{aligned}
		\end{equation}
		\begin{equation}
			X^R_{t,0} =  e^{-r(T-t)} (a_0-A_0) \frac{\Phi'\left(d_1\left(\frac{\gamma^+_0}{y^*\xi_t}\right)\right)}{\Phi'\left(d^0\left(\frac{\gamma^+_0}{y^*\xi_t}\right)\right)}\left[ \Phi\left(d^0\left(\frac{\gamma^-_1}{y^*\xi_t}\right)\right) -\Phi\left(d^0\left(\frac{\gamma^+_0}{y^*\xi_t}\right)\right) \right] \id_{\left\{ R_0 \neq 0,\infty \right\}} \rightarrow 0.
		\end{equation}
		Hence, in this case,
		\begin{equation}\label{eq:gamma0_case1}
			X_t^* = \sum_{k = 0}^n \( X_{t,k}^D + X_{t,k}^A + X_{t,k}^R \) \rightarrow e^{-r(T-t)} a_0.
		\end{equation}
		If $\gamma_0^+ =\infty$, then $a_0 = A_0$, which implies $X_{t,0}^R = 0$ and 
		\begin{equation}
			X_t^* = \sum_{k = 0}^n \( X_{t,k}^D + X_{t,k}^A\) + \sum_{k = 1}^n X_{t,k}^R.
		\end{equation}
		As $\xi_t \rightarrow \infty$, we have
		\begin{equation}
			\begin{aligned}
				X^D_{t,0} &= e^{-r(T-t)} a_0  \left[ \Phi\left(d_1\left(\frac{\gamma^+_0}{y^*\xi_t}\right)\right) -\Phi\left(d_1\left(\frac{\gamma^-_0}{y^*\xi_t}\right)\right) \right] 
				= 0-0 = 0,
			\end{aligned}
		\end{equation}	
		\begin{equation}
			\begin{aligned}
				X^A_{t,0} &=  e^{-r(T-t)} A_0 \left[ \Phi\left(d_1\left(\frac{\gamma^-_1}{y^*\xi_t}\right)\right) -\Phi\left(d_1\left(\frac{\gamma^+_0}{y^*\xi_t}\right)\right) \right] \id_{\left\{ R_0 \neq 0,\infty \right\}} \rightarrow e^{-r(T-t)} A_0 = e^{-r(T-t)} a_0.
			\end{aligned}
		\end{equation}
		Hence, in this case, similarly to \eqref{eq:gamma0_case1}, we also have
		\begin{equation}\label{eq:gamma0_case2}
			X_t^* = \sum_{k = 0}^n \( X_{t,k}^D + X_{t,k}^A \)  + \sum_{k = 1}^n X_{t,k}^R \rightarrow e^{-r(T-t)} a_0.
		\end{equation}
		Moreover, according to the expressions of \eqref{eq:pkqk}, for any $k \in \{1, \cdots, n\}$, as $\xi_t \rightarrow \infty$, we have 
		\begin{equation}
			q_k \rightarrow 0, \;\;     p_k \rightarrow 0.
		\end{equation}
		For $k = 0$, as $\xi_t \rightarrow \infty$, because $\gamma_0^- =\infty$, we have
		\begin{equation}
			q_0 = \Phi\(d_1 \( \frac{\gamma_1^-}{y^* \xi_t} \) \) - \Phi\(d_1 \( \frac{\gamma_0^+}{y^* \xi_t} \)\)  
			\rightarrow \left\{
			\begin{aligned}    
				& 1-0 = 1, && \text{ if } \gamma_0^+ = \infty;\\ 
				& 1-1 = 0, && \text{ if } \gamma_0^+ \in [0, \infty), 
			\end{aligned}
			\right.   
		\end{equation}
		\begin{equation}
			p_0 = \Phi\(d_1 \( \frac{\gamma_0^+}{y^* \xi_t} \) \) - \Phi\(d_1 \( \frac{\gamma_0^-}{y^* \xi_t} \)\) = \Phi\(d_1 \( \frac{\gamma_0^+}{y^* \xi_t} \) \) 
			\rightarrow \left\{
			\begin{aligned}    
				& 0, && \text{ if } \gamma_0^+ = \infty;\\ 
				& 1, && \text{ if } \gamma_0^+ \in [0, \infty). 
			\end{aligned}
			\right.   
		\end{equation}
		Thus, we obtain that if $\xi_t \rightarrow \infty$, then
		\begin{equation}
			\boldsymbol{\pi}_t^{(4)} \rightarrow 
			\left\{
			\begin{aligned}
				& \mathbf{0}, && \text{ if } \gamma_0^+ = \infty;\\
				& -\frac{e^{-r(T-t)}}{R} (\boldsymbol{\sigma}^\top)^{-1} \boldsymbol{\theta} a_0, && \text{ if } \gamma_0^+ \in [0, \infty),
			\end{aligned}
			\right.
		\end{equation}
		\begin{equation}
			\boldsymbol{\pi}_t^{(3)} \rightarrow 
			\left\{
			\begin{aligned}
				& -\frac{e^{-r(T-t)}}{R} (\boldsymbol{\sigma}^\top)^{-1} \boldsymbol{\theta} a_0, && \text{ if } \gamma_0^+ = \infty;\\
				& \mathbf{0}, && \text{ if } \gamma_0^+ \in [0, \infty),
			\end{aligned}
			\right.
		\end{equation}
		\begin{equation}
			\boldsymbol{\pi}_t^{(2)} \rightarrow \mathbf{0},\;\; 
			\boldsymbol{\pi}_t^{(1)} \rightarrow \frac{1}{R} (\boldsymbol{\sigma}^\top)^{-1} \boldsymbol{\theta} e^{-r (T-t)} a_0.
		\end{equation}
		Hence, as $\xi_t^* \rightarrow \infty$, we have
		$$
		\boldsymbol{\pi}_t^* \rightarrow \mathbf{0}.
		$$
	\end{enumerate}
\end{proof}

\end{document}